\documentclass[11pt]{article}
\usepackage[margin=0.8in]{geometry}

\usepackage{booktabs} %
\usepackage[ruled,linesnumbered]{algorithm2e} %

\usepackage{enumitem}
\usepackage[hidelinks]{hyperref}

\usepackage{amsmath}
\usepackage{amssymb}
\usepackage{mathtools}
\usepackage{amsthm}
\usepackage{thmtools, thm-restate}
\usepackage{bm}
\usepackage{xcolor}
\usepackage{natbib}
\renewcommand{\cite}{\citep}

\theoremstyle{plain}
\newtheorem{theorem}{Theorem}[section]
\newtheorem{proposition}[theorem]{Proposition}
\newtheorem{lemma}[theorem]{Lemma}

\theoremstyle{definition}
\newtheorem{example}{Example}
\newtheorem{definition}[theorem]{Definition}
\newtheorem{assumption}[theorem]{Assumption}
\newtheorem{observation}[theorem]{Observation}
\theoremstyle{remark}

\usepackage{xcolor}
\newcommand{\cF}{\mathcal{R}}
\newcommand{\Q}{Q}
\newcommand{\R}{\mathbb{R}}
\newcommand{\cR}{\mathcal{R}}

\newcommand{\V}{\mathcal{V}}

\newcommand{\E}[2]{\mathbb{E}_{#1}\left[#2\right]}

\newcommand{\Mat}{\mathrm{Mat}}

\newcommand{\Unif}[1]{\mathrm{Unif}\left(#1\right)}
\newcommand{\eps}{\varepsilon}
\newcommand{\sph}{\mathbb{S}}
\newcommand{\sdm}{\sph^{d-1}}
\newcommand{\Lip}{\text{Lip}}

\newcommand{\ip}[2]{\langle #1,#2\rangle}

\newcommand{\norm}[1]{\left\| #1 \right\|}
\newcommand{\tp}{^\top}

\newcommand{\defeq}{\coloneqq}
\newcommand{\moms}{\bm{m}}
\newcommand{\Moms}{\bm{M}}
\newcommand{\bq}{\bm{q}}

\newcommand{\resp}{\mathrm{resp}}

\usepackage{dsfont}
\DeclareMathOperator{\1}{\mathds{1}}
\newcommand\abs[1]{\lvert #1 \rvert}

\DeclareMathOperator{\vrw}{\mathrm{raw}_\alpha}
\DeclareMathOperator{\raw}{\vrw}
\DeclareMathOperator{\tcw}{\mathrm{tcw}}
\DeclareMathOperator{\tc}{\mathrm{tc}}

\usepackage[capitalize,noabbrev]{cleveref}

\title{Linear Social Choice with Few Queries: A Moment-Based Approach}

\author{Luise Ge\thanks{Washington University in St. Louis. {\tt g.luise@wustl.edu}} \and Daniel Halpern\thanks{Google Research. {\tt dhalpern@google.com}} \and Gregory Kehne\thanks{Washington University in St. Louis. {\tt kehne@wustl.edu}} \and Yevgeniy Vorobeychik\thanks{Washington University in St. Louis. {\tt yvorobeychik@wustl.edu}}}

\date{}

\begin{document}

\maketitle

\begin{abstract}
  Most social choice rules assume access to full rankings, while current alignment practice—despite aiming for diversity—typically treats voters as anonymous and comparisons as independent, effectively extracting only about one bit per voter. Motivated by this gap, we study social choice under an extreme communication budget in the linear social choice model, where each voter’s utility is the inner product between a latent voter type and the embedding of the context and candidate. The candidate and voter spaces may be very large or even infinite. Our core idea is to model the electorate as an unknown distribution over voter types and to recover its moments as informative summary statistics for candidate selection. We show that one pairwise comparison per voter already suffices to select a candidate that maximizes social welfare, but this elicitation cannot identify the second moment and therefore cannot support objectives that account for inequality. We prove that two pairwise comparisons per voter, or alternatively a single graded comparison, identify the second moment; moreover, these richer queries suffice to identify all moments, and hence the entire voter-type distribution. These results enable principled solutions to a range of social choice objectives including inequality-aware welfare criteria such as
taking into account the spread of voter utilities and choosing a representative subset.
\end{abstract}

\maketitle

\section{Introduction}
\label{sec:intro}

A fundamental problem in social choice is to map potentially diverse preferences of a collection of voters over a set of candidates to a subset of \emph{winners}.
The classical model assumes that voter preferences are elicited as full rankings. 
In many real settings, however, the voter population can be massive and the candidate space is enormous or unbounded—for example, voters may be users of an AI system (such as an LLM) and candidates may be possible outputs like images, music, recipes, or answers to complex prompts. Moreover, preferences over candidates may depend a great deal on context, which itself can defy enumeration: for example, preferences over responses clearly depend on the question.
This situation has become particularly salient in the context of AI \emph{value alignment}, such as training large language models (LLMs) to align with (i.e., behave according to) humans' subjective values such as helpfulness and harmlessness~\cite{bai2022training,ji2023beavertails,ouyang2022training,Ouyang2025allocation}.
Since the space of candidates cannot be explicitly enumerated, a typical paradigm would present specific pairs of candidates to human annotators, who would select which of these they prefer, or indicate indifference. Such datasets are typically collected without retaining annotator identifiers; hence, the feedback per voter can be as sparse as one bit. Approaches such as reinforcement learning from human feedback (RLHF) then use these pairwise comparisons to first train a parametric reward (utility) function.
This reward model is subsequently plugged into a conventional RL framework such as PPO~\cite{schulman2017proximal} to achieve context-dependent behavioral alignment~\cite{christiano2017deep,stiennon2020learning}.

While there has been an increasing interest in incomplete vote elicitation~\cite{halpern2023representation,halpern2024computing}, the practice of AI alignment has pushed the feedback constraint to an extreme. %
On the one hand, if we wish to achieve useful guarantees about selecting good candidates in such settings (e.g., in terms of social welfare if voter preferences reflect latent utility functions), the situation seems hopeless.
On the other hand, it is typical in settings such as value alignment that the space of candidates and voters is \emph{structured}.
In particular, generative AI methods rely on embedding digital objects such as music, images, or text, as vectors.
It is then natural to posit that human preferences, too, have a structured representation as parametric utility functions over these vectors, an assumption that is exploited in the reward model learning step of RLHF.
This motivates our central research question:
\begin{quote}
\emph{Is it possible to leverage vector representations of candidates and parametric representations of voter utilities to obtain sufficiently reliable information from few per-voter pairwise comparison queries to select good winning candidates, when voter and candidate spaces are both large?}
\end{quote}

To study this question, we assume that each voter has a utility function over the vector space of candidates which is linear in the candidate embedding.
Motivated by the linear representation hypothesis~\cite{linear-rep}, this effectively assumes that LLM text embeddings are rich enough that user preferences can be represented as linear functions over them. 
Yet even in this structured setting, the initial results have been bleak. Even with access to full rankings, many aggregation rules including the standard RLHF procedure fail the most basic social choice properties like Pareto Optimality~\cite{ge2024axioms}. Moreover, it is information-theoretically impossible to identify a candidate to achieve social welfare within a constant factor of optimal~\cite{ge2025optimized}.

Crucially, these negative conclusions are obtained in a finite-electorate setting with a fixed dataset of comparisons. Our focus is instead on preference elicitation under limited communication. We assume access to a large population that can be sampled repeatedly, and we can choose the comparisons we ask, but each voter provides only a small number of binary responses. 
Still, suppose that we aim to choose a single candidate to (approximately) maximize social welfare.
How many pairwise rankings do we need to elicit from each voter?
Moreover, assuming we cannot ask all possible voters (e.g., the entire population of a country), how many voters suffice, if we treat each as a sample from an unknown voter distribution?
Furthermore, how does this picture change as we ask more complex questions, such as (a) maximizing welfare while accounting for inequality aversion, or (b) selecting a committee rather than a single candidate?

\subsection{Our Contributions}

We cast our setting as social choice under sparse elicitation: each voter has a latent preference vector $\theta \in \sdm$, but the mechanism can ask only a small number of comparison-style queries per voter. This raises three basic questions: (i) which social objectives are meaningful in this limited-information regime, (ii) what information should be elicited to evaluate those objectives, and (iii) how many samples (voters) are required.

Our answers are organized around a single principle: moments are appropriate summaries of the preference distribution. We show how different query families identify different moments, and we give finite-sample guarantees that translate moment estimation into guarantees for downstream social-choice objectives.

Concretely, first consider selecting a single candidate. If the goal is to maximize social welfare, it turns out only a single such query suffices.
That is, we show that we can both \emph{identify} (Section~\ref{sec:identifiability}) and \emph{effectively estimate} (Section~\ref{sec:estimation}) the first moment of the voter distribution \emph{with only a single pairwise comparison query per voter}:
\begin{theorem}[informal]
The first moment of the voter distribution is identifiable using one pairwise comparison query per voter. Moreover, we can estimate it to within $\eps$ with sample complexity polynomial in $1/\eps$ and $d$.
\end{theorem}
As a direct consequence of this result, \emph{we can find an approximately welfare-optimal candidate by asking each voter only one pairwise comparison query} (Section~\ref{sec:apps}).

While maximizing welfare is a natural goal in social choice, it has a significant limitation: the result can be extremely inequitable, for example, with some voters having very high, while many others very low, utility over the final candidate.
It is often desirable to moderate this criterion by using welfare objectives that also account for inequality~\cite{atkinson1970measurement}.
fOne measure of inequality is utility variance, which we can leverage to construct welfare functions that combine average utility with variance.
In order to identify a candidate that maximizes a variance-adjusted welfare function, a crucial subproblem is to estimate \emph{the second moment of the voter distribution}.
Is a single pairwise comparison query to each voter sufficient for this?
We show that it is not (Section~\ref{sec:identifiability}):

\begin{theorem}[informal]
The second moment of the voter distribution is not identifiable from one pairwise comparison query per voter.
\end{theorem}
Thus, we need more than a single pairwise comparison per voter to obtain variance information.
Do two such comparisons suffice?
We show that they do---in fact, we show that this generalizes directly for any $k$ (Section~\ref{sec:identifiability} shows identifiability while Section~\ref{sec:estimation} considers estimation):
\begin{theorem}[informal]
The first $k$ moments of the voter distribution are identifiable using $k$ pairwise comparison queries per voter, and can be estimated from such data to within $\eps$ with sample complexity polynomial in $1/\eps$ and $d$.
\end{theorem}
Intuitively, it seems that $k$ pairwise queries are also \emph{necessary} to identify $k$ moments of a distribution. Remarkably, this is not true: only 2 per voter queries suffice to estimate \emph{any properties of the distribution} (including arbitrary moments):
\begin{theorem}[informal]
The voter distribution is identifiable using $2$ pairwise comparison queries per voter, and the $k$th moment can be estimated from such data to within $\eps$ with sample complexity polynomial in $1/\eps$ and $d$.
\end{theorem}

We can also significantly generalize the above results to stochastic response models; details are provided in \Cref{app:generalized-responses}. Furthermore, if voters only respond if they \emph{strongly} prefer one response to another (formally defined in Section~\ref{subsec:graded-query}), then only one query per voter is necessary to identify the distribution.

Our moment-based approach is also useful in selecting candidates under other objectives. For example, we show how moments can be used to approximate Nash welfare. We also apply this to selecting committees of candidates in multi-winner elections.
In this context, we need to extend the voter utility functions to \emph{sets} of candidates.
For example, we can assume that each voter's utility of a set is the maximum utility from any candidate in the set.
By approximating this using a $k$-degree polynomial and maximizing the resulting function, we can obtain an approximately social-welfare-optimal committee.

\subsection{Related Work}
\label{sec:relwork}

Our work operates within the \emph{Linear Social Choice} framework, a new paradigm of social choice in which voter utilities are linear functions of context and candidate embeddings.
Prior results emphasize worst-case aspects from axiom violations~\cite{ge2024axioms} to distortion bounds~\cite{ge2025optimized}. In contrast, we take a bottom-up perspective and ask information-theoretic questions. There is a rich literature on dealing with incomplete information in social choice~\citep[Chapter 10]{brandt2016handbook}. Our work is most closely related to recent frameworks that combine elicitation with distributional assumptions, where only minimal information is elicited from each voter~\cite{halpern2023representation,halpern2024computing}. These works, however, focus on different feedback models (e.g., approval) and different questions (e.g., computing voting rules), and do not operate in the linear-utility setting.

Our contributions also intersect with the extensive preference learning literature whose roots lie in ~\emph{learning to rank}~\cite{cohen1997learning,burges2005learning}, and which has been adopted for \emph{virtual democracy} ~\cite{noothigattu2018voting,kahng2019statistical}. More recently, learning a parametric utility model has been analyzed theoretically \cite{Ge2024Learning} and used in practice as part of RLHF-style training~\cite{christiano2017deep} to enable generalization over large candidate sets. Notably, these works deal with a single preference, and observed disagreement is treated as noise. 

In contrast, pluralistic alignment and the increasing use of LLM-based auto-raters~\cite{Li2025Judging} shift the focus toward learning the voter population itself. 
Efforts in this direction are so far predominantly empirical~\cite{chakraborty2024maxmin, melnyk2024distributional,siththaranjan2024distribution,kim2025population}. 
The more theoretical works to date~\cite{chidambaram2025direct,cherapanamjeri2025learning,shirali2025direct} adopt different models but, broadly, resonate with our findings: the standard alignment data pipeline is insufficient for learning heterogeneity in the underlying population.

Our estimation strategy relies on the \emph{Generalized Method of Moments} (GMM)~\cite{hansen1982large,pearson1936method}. While GMM has been applied to estimate parameters of probabilistic ranking models like Plackett-Luce~\cite{azari2013generalized}, we adapt the principle to the linear setting. At a technical level, our problem is also related to \emph{1-Bit Compressed Sensing} in signal processing~\cite{boufounos20081,plan2013one}, which recovers signals from the signs of random linear measurements. Our setting shares this structure of receiving single bits of information per sample (comparison queries), though we focus on recovering distributional moments rather than sparse vectors, and allowing multiple comparison queries at once. At a more fundamental level, our identifiability results connect to \emph{geometric tomography} and the \emph{Cramér-Wold theorem}~\cite{cramer1936some}, which states that a high-dimensional distribution is determined by its lower-dimensional projections. We extend these insights to show how ``projections'' obtained via pairwise queries are sufficient to recover properties of the voter distribution.

\subsection{Organization}

The rest of the paper is organized as follows.
First, we provide formal preliminaries in Section~\ref{sec:prelims}.
Our main results regarding identifiability of moments of the voter distribution from pairwise comparison queries then follow in Section~\ref{sec:identifiability}.
Next, we build on the positive identifiability results to characterize how to estimate moments in Section~\ref{sec:estimation}.
Finally, we illustrate our results on moment estimation in the context of social choice applications in both single-winner and committee selection in Section~\ref{sec:apps}.

\section{Preliminaries}
\label{sec:prelims}

\subsection{Voters and Utilities}
We assume that voters $v\in V$ are distributed according to a distribution $\V$ which represents our underlying voter population. 
Each voter is characterized by a type vector $\theta_v$ lying on the unit sphere $\mathbb{S}^{d-1} \defeq \{ \theta \in \mathbb{R}^d : \|\theta\|_2 = 1 \}$ to control for utility scale invariance of preference rankings. This mapping induces a voter type distribution $\Theta$ over $\mathbb{S}^{d-1}$. For ease of exposition, we assume that $\Theta$ is absolutely continuous with respect to the Lebesgue measure on the sphere.\footnote{Informally, this implies that the probability mass of the voter distribution is not concentrated on lower-dimensional subsets. However, with additional technical care, our results can be extended to hold without this assumption.}

We consider a space of contexts (prompts) $\mathcal{X}$ and a space of candidates (responses) $\mathcal{Y}$. An embedding function $\Phi: \mathcal{X} \times \mathcal{Y} \to \mathbb{R}^d$ maps each context-candidate pair to a real vector. We assume that voter utilities are linear in these embeddings; specifically, the utility of a voter with type $\theta$ for a pair $(x, y)$ is given by the inner product $u_\theta(x, y) = \theta \cdot \Phi(x, y)$. Since utilities are fully determined by voter types and candidate embeddings, we will henceforth identify voters directly with their types $\theta$ and refer to $\Theta$ as the voter distribution. Similarly, we will often refer to a context-candidate pair $(x, y)$ by its embedding $\phi = \Phi(x, y)$, writing the utility function simply as $u_\theta(\phi) = \theta \cdot \phi$. Finally, we will write $\mathcal{U}_\phi$ for the distribution of utilities induced by $\phi$, i.e., $\mathcal{U}_\phi$ is the distribution induced by $u_\theta(\phi)$ over the randomness of $\theta$. A reference table is provided in Appendix~\ref{app:notation}.

\subsection{Welfare Objectives}

The first objective to consider is \emph{social welfare}, which is the expected utility w.r.t.\@ the voter distribution $\E{\theta \sim \Theta}{u_\theta(\phi)}$.
However, simply maximizing welfare can yield considerable inequality in realized utility across voters.
A number of alternative welfare notions therefore aim to adjust for potential inequality.
A natural way to do this is maximizing 
\emph{risk-adjusted} welfare (raw), maximize the expected welfare but penalize a candidate by $\alpha$ times the standard deviation:
\begin{definition}[$\alpha$-risk-adjusted welfare]
\label{def:variance-reg-welf-objective}
    For $\alpha \geq 0$, the \emph{$\alpha$-risk-adjusted welfare} of $\phi$ is given by 
    \begin{equation*}
        \vrw(\phi) \defeq \E{\Theta}{u_\theta(\phi)} - \alpha \cdot \sqrt{\E{\Theta}{\left(u_\theta(\phi) - \E{\Theta}{u_\theta(\phi)} \right)^2}}.
    \end{equation*}
\end{definition}
Another common objective that has the effect of being more equitable than social welfare is \emph{Nash welfare}, which in our setting is defined as $\E{\theta \sim \Theta}{\log u_\theta(\phi)}$.
Of course, for this to be well-defined and meaningful, voter utilities $u_\theta(\phi)$ must be strictly positive.

Alternatively, it is often possible and desirably to compute a set of winning candidates (that is, a \emph{committee}) $W$ rather than a single winner.
In a welfarist context, we need to extend a voter's utility over individual candidates to a utility function over sets.
One natural extension is that the voter's utility for $W$ stems from their most preferred candidate in $W$, i.e., $u_\theta(W) = \max_{\phi \in W} u_\theta(\phi)$.
This yields a welfare objective that we refer to as \emph{top-choice welfare}:

\begin{definition}[Top-choice welfare]
\label{def:top-choice-welfare}
    For a user distribution $\Theta$ and a set of candidate responses $\Phi$, the \emph{top-choice welfare} of a set of candidates $W \subseteq \Phi$ is given by 
    \[
        \tcw_\Theta(W) \defeq \E{\Theta}{\max_{\phi \in W} u_\theta(\phi)}.
    \]
    We refer to the problem of optimizing $\tcw$ subject to $\abs{W} \leq \ell$ for a given $\ell$ as \emph{$\ell$-$\tcw$ maximization}.
\end{definition}

\subsection{Preference Elicitation}
As is typical in the value alignment literature, voters do not report their utilities directly (it is typically too much to ask).
Instead, we elicit preference information via pairwise comparison queries. A single query consists of a context $x$ and a pair of candidate responses $y_1, y_2$. When presented with such a query, a voter $\theta$ indicates which response yields higher utility, returning $1$ if they prefer $y_1$ and $0$ if they prefer $y_2$.\footnote{Our framework can be naturally extended to incorporate stochastic responses; see Appendix~\ref{app:generalized-responses}.} We encode this response as $\text{resp}_\theta((x, y_1, y_2)) = \mathbb{I}\{u_\theta(x, y_1) \ge u_\theta(x, y_2)\}$.\footnote{This formulation effectively breaks ties in favor of $y_1$, but the specific tie-breaking mechanism does not impact our results. By the absolute continuity of $\Theta$, the probability that a random voter assigns exactly equal utility to any two distinct candidate embeddings is zero.}

In our linear utility model, the condition $u_\theta(x, y_1) \ge u_\theta(x, y_2)$ is equivalent to $\theta \cdot (\Phi(x, y_1) - \Phi(x, y_2)) \ge 0$. Since the response depends solely on the \emph{difference} between the embeddings, we define the query vector $q \defeq \Phi(x, y_1) - \Phi(x, y_2)$ and allow the response function to operate directly on these differences, i.e., $\text{resp}_\theta(q) \defeq \1\{\theta^\top q \geq 0\}$.

Our results rely on a geometric assumption that for any direction, we can identify a context and pair of candidates that (approximately) induce this vector by the difference in the associated embeddings.
In effect, this means that the space of contexts and candidates must be sufficiently rich, at least with respect to the induced embedding space.
This assumption allows us to establish the fundamental limits of this form of preference elicitation. Impossibility results in our setting imply impossibility under any weaker query model. Conversely, positive identifiability results establish a theoretical ceiling, reducing the alignment problem to the engineering challenge of generation.

\begin{assumption}\label{assum:express}
     The embedding space is sufficiently expressive such that for any direction $q \in \sdm$, we can generate $(x, y_1, y_2)$ such that $\Phi(x, y_1) - \Phi(x, y_2) \propto q$.
\end{assumption}
Under Assumption~\ref{assum:express}, the problem of selecting a comparison $(x, y_1, y_2)$ reduces to directly choosing a direction $q \in \sdm$. We refer to such a vector $q$ as a \emph{query} and treat these queries as the primary decision variables for the data collector.

\subsection{Multi-Query Responses}

Given a sequence of $t$ queries $\mathbf{q} = (q_1, \ldots, q_t)$, a random voter $\theta \sim \Theta$ arrives and provides a response vector $(\resp_\theta(q_1), \ldots, \resp_\theta(q_1)) \in \{0, 1\}^t$ indicating their preferences. Consequently, a fixed query sequence $\mathbf{q}$ induces a distribution over binary response vectors. We define $Q_t(\mathbf{q}, \mathbf{b}; \Theta)$ as the probability that a random voter drawn from $\Theta$ produces the response vector $\mathbf{b} \in \{0, 1\}^t$ when presented with queries $\mathbf{q}$. When the underlying distribution $\Theta$ is clear from context, we will omit it from the notation. Thus, for any fixed $\mathbf{q}$, the function $Q_t(\mathbf{q}, \cdot)$ represents a probability distribution over $\{0, 1\}^t$.

Throughout, we will use the shorthand $\Q_t(\bq)$ to denote the probability of the all-positive response vector, $\Q_t(\bq, \bm{1})$, where $\bm{1} = (1, \ldots, 1)$. 
Formally, we have
\begin{equation*}
\label{eq:k-comparison-query-def}
    \Q_t( q_1,\dots,q_t) \defeq \Pr_{\theta\sim\Theta}\left[[\resp_\theta(q_1) = 1] \wedge \ldots \wedge [\resp_\theta(q_t) = 1]\right].
\end{equation*}

It is worth noting that the probability of any arbitrary response pattern $\mathbf{b} \in \{0,1\}^t$ can be recovered solely from the values of $Q_t(\mathbf{q})$ (for instance, by negating specific query vectors $q_i$ to target zeros, since $\text{resp}_\theta(-q) = 1 - \text{resp}_\theta(q)$ almost surely). 

Finally, we remark that we assume deterministic voter responses in the main paper for clarity of presentation. Nevertheless, our results for multi-query elicitation extend directly to stochastic response models (e.g., Bradley--Terry); see \Cref{app:generalized-responses}.

\subsection{Graded-Query Responses} \label{subsec:graded-query}
Beyond ordinal comparisons, recent work shows that even a few bits of cardinal utility can improve distortion in single-winner elections~\cite{amanatidis2021peeking,ebadian2025every}. But reporting cardinal values remains cognitively challenging. By contrast, reporting preference intensity (e.g., “strong” versus “weak”) is often less demanding and is already collected in many preference datasets. We therefore follow the line of work that incorporates preference intensity~\cite{kahng2023voting}. Concretely, we model a graded preference query as whether the utility margin exceeds a threshold $\tau \in (0,1)$: $\mathrm{grad}_{\theta}(q)\defeq\1 \{\theta \tp q \ge \tau \}$. Then for a distribution $\Theta$, $G_{\tau}(q) = \Pr_{\theta\sim\Theta}[\mathrm{grad}_{\theta}(q)]$ returns the fraction of voters having a $\tau$-strong preference.

Since the inner product $\theta \tp q$ depends on the norm of $q$, for results on graded queries, we need to further assume that the query space is rich enough to produce $q \in \sdm.$
\begin{assumption}\label{assum:express_2}
     The embedding space is sufficiently expressive such that for any direction $q \in \sdm$, we can generate $(x, y_1, y_2)$ such that $\Phi(x, y_1) - \Phi(x, y_2) = q$.
\end{assumption}

\subsection{Distributions and their Moments}

We adopt the standard measure-theoretic formalism of probability, in which probability distributions are defined over a sample space $\Omega$ with an associated $\sigma$-algebra $\cF$ of measurable events.
Probability measures are then measures $\mu$ on $(\Omega, \cF)$ that are normalized such that $\mu(\Omega) = 1$. 
We will use $\bar{\sigma}$ to describe the uniform probability measure over the sphere $\sdm$, to differentiate it from $\sigma$ which is the \emph{unnormalized} surface area measure over $\sdm$. 

Any measure $\mu$ over $\R$ has an associated sequence of moments $\moms = (m_0, m_1, m_2, \ldots)$ given by $m_n \defeq \int_\R x^n \:d\mu(x)$.
Under certain assumptions, as in the Hausdorff moment problem, where $\mu$ is restricted to $[0,1]$, they are sufficient to uniquely specify $\mu$ (up to sets of measure $0$).

We use tensors to represent the higher-order moments. Formally, a $k$th order tensor $T \in (\mathbb{R}^d)^{\otimes k}$ is a multidimensional array indexed by a tuple $(i_1, \dots, i_k)$ where each $i_j \in [d]$.
For vectors $v_1, \dots, v_k \in \mathbb{R}^d$, the \emph{outer product} $v_1 \otimes \dots \otimes v_k$ is a tensor with entries given by the product of the coordinates:
$
    (v_1 \otimes \dots \otimes v_k)_{i_1, \dots, i_k} = (v_1)_{i_1} \cdot (v_2)_{i_2} \cdots (v_k)_{i_k}.
$
When $k=1$, this is equivalent to the vector itself; when $k=2$, this corresponds to the matrix outer product $v_1 v_2^\top$.
For two tensors $A, B \in (\mathbb{R}^d)^{\otimes k}$, their inner product is the sum of the products of their corresponding entries:
$    \langle A, B \rangle \defeq \sum_{i_1, \dots, i_k = 1}^d A_{i_1, \dots, i_k} B_{i_1, \dots, i_k}.
$
We sometimes reshape tensors into matrices. For a partition of the modes $\{1, \dots, k\}$ into two sets $I$ and $J$, the \textit{matricization} $\text{mat}_{I|J}(T)$ flattens the tensor into a matrix where the rows are indexed by $I$ and the columns by $J$.  And sometimes we symmetrize a tensor of rank $k$ as $\mathrm{Sym}(T)\defeq \frac{1}{k!}\sum_{\pi \in S_k} T^{\pi}$, with $S_k$ the symmetric group and $T^{\pi}$ the tensor re-indexed according to the permutation $\pi$.

For a vector-valued distribution $\Theta$ over $\R^d$, the moments $\Moms = (M_0, M_1, M_2, \ldots)$ can then be expressed using $k$-tensors:
\begin{equation*}
\label{def:vector-moments}
     M_k \defeq \E{\theta \sim \Theta}{\theta^{\otimes k}} = \int_{\R^d} \theta^{\otimes k} \: d\mu(\theta).
\end{equation*}
Intuitively, this is a $k$-dimensional array with entries indexed by sequences of $k$ indices, $i_1, \ldots,i_k \in [d]$, where entries are the scalars $(M_k)_{i_1,\dots,i_k} = \E{\theta \sim \Theta}{\theta_{i_1}\cdots \theta_{i_k}}$.
We will write $\Moms(\Theta)$ and $M_k(\Theta)$ when the distribution (measure) is not clear from context.

We can now relate the moments of the distribution over utilities $\mathcal{U}_\phi$ for a candidate $\phi$ to the moments of $\theta$. In particular, let $m_k = \E{\theta \sim \Theta}{u_\theta(\phi)^k}$ be the $k$th moment.
Then by tensor arithmetic and linearity of $u_\theta(\phi)$, $m_k = \langle M_k, \phi^{\otimes k} \rangle$.

\subsection{Moment Identifiability and Estimation}
\label{sec:model-identifiability}

Our first concern is information-theoretic: what properties $\mathcal{P}$ of the voter distribution $\Theta$ (e.g., its $k$th moment tensor $M_k(\Theta)$) can be determined by the responses? We say that $\mathcal{P}$ is \emph{identifiable} from $Q_t$ if, for any two distributions $\Theta, \Theta'$ that yield identical query responses (e.g., $Q_t(\Theta) = Q_t(\Theta')$), it holds that $\mathcal{P}(\Theta) = \mathcal{P}(\Theta')$.
For example, the $k$th moment is identifiable with $t$-sized queries if, for any two distributions with distinct $k$th moments, there exist a $\mathbf{q} = (q_1, \ldots, q_t)$ and $\mathbf{b} = (b_1, \ldots, b_t)$ such that $Q_t(\mathbf{q}, \mathbf{b}; \Theta_1) \ne Q_t(\mathbf{q}, \mathbf{b}; \Theta_2)$.

In applications, we must \emph{estimate} moments from finitely many voters. Since we are ultimately interested in using moments to estimate the utility distribution of given $\phi$, we measure estimation error using the \emph{spectral norm}, which for a $k$th order tensor $T \in (\mathbb{R}^d)^{\otimes k}$ is defined as $$\|T\| \defeq \sup_{u_1, \dots, u_k \in \mathbb{S}^{d-1}} \left| \left\langle T, u_1 \otimes \cdots \otimes u_k \right\rangle \right|.$$
For $k = 1$, the spectral norm corresponds to the standard $L_2$ norm. Control over the spectral norm guarantees uniform accuracy in estimating the moments of the utility distribution. 
Specifically, if $\widehat{M}_k$ satisfies $\|\widehat{M}_k - M_k\| \le \eps$, then for any $\phi \in \mathbb{R}^d$, the estimated $k$th moment of the utility distribution satisfies:
\(
    \left| \left\langle \widehat{M}_k, \phi^{\otimes k} \right\rangle - \E{\theta \sim \Theta}{\left\langle\theta, \phi\right\rangle^k} \right| \le \eps \|\phi\|^k.
\)

\section{Identifiability of Moments from Queries}
\label{sec:identifiability}

A key first step toward estimating moments from pairwise comparisons is understanding the \textit{minimum} number of queries per voter needed for identifiability. In this section, we provide such results in the form of identifiability. We start by considering the first moment $M_1$, which in our setting translates to effective social welfare maximization, showing that only a single per-voter query suffices.
Next, we show several more general results: 1) we can identify the first $k$ moments using $k$ per-voter queries, and 2) only two queries or only a single graded query per voter are sufficient to identify all moments (and hence, the distribution).

\subsection{Identifying the Average Voter}
\label{sec:1-moment}
This first moment $M_1$ is the average voter, i.e., $M_1 \defeq \bar \theta = \E{\theta \sim \Theta}{\theta}$.
\begin{observation}
\label{obs:welfare-max-from-first-moment}
    Knowing $\bar \theta = M_1$ suffices to maximize welfare of $\Theta$; in particular, for any candidate embedding $\phi$, we have $\E{\theta \sim \Theta}{u_\theta(\phi)} = \E{\theta \sim \Theta}{\theta\tp \phi} = \bar \theta \tp \phi$.
\end{observation}

We now show that $M_1$ can be identified from $1$-sized queries alone. 
Recall that $\Q_1$ denotes pairwise comparison queries.
Consider the contribution of a voter type $\theta$ to $\Q_1$ when averaged over all query directions $q$, which we denote by
\begin{equation*}
\label{def:avg-queried-direction-fxn}
    I(\theta) \defeq \int_{\sdm} \resp_\theta(q) \cdot q \: d\bar{\sigma}(q).
\end{equation*}

As query directions $q$ are uniform, symmetry implies that this is itself in the direction of $\theta$.
\begin{lemma}
\label{lem:first_identity}
    For any $\theta\in\mathbb S^{d-1}$ it holds that $I(\theta) = c_d \cdot \theta$, where
    \begin{equation}
    \label{eq:cd-constant-def}
       c_d = \frac{\Gamma(d/2)}{2\sqrt{\pi}\Gamma(\frac{d+1}{2})} =  \Theta(d^{-1/2}).
    \end{equation}
    In particular, $c_d \ge \frac{1}{\sqrt{2\pi}} \cdot d^{-1/2}$. 
    
\end{lemma}
\begin{proof}[Proof of \Cref{lem:first_identity}]
    Let $R$ be any rotation in $\mathbb R^d$.  
    Because $\langle R\theta,q\rangle = \langle\theta,R^{-1}q\rangle$ and
    $d\bar{\sigma}$ is rotation invariant,
    \begin{align*}
        I(R\theta)
       =& \int \1\{\ip{R\theta}{q}\ge 0\} \cdot q \:d\bar{\sigma}(q)
       = \int \1 \{\ip{\theta}{R^{-1}q}\ge 0\} \cdot q \:d\bar{\sigma}(q) \\
       = & \int \1\{\ip{\theta}{q'}\ge 0\}\, Rq'\, d\bar{\sigma}(q')
       = R\, I(\theta).
    \end{align*} 
    Let $\mathcal R_\theta$ be the subgroup of rotations of $\R^d$ that fix $\theta$.
    For any $R\in \cR_\theta$, we have $R\theta=\theta$, and
    therefore by the equivariance established above,
    \(
       R\,I(\theta) = I(R\theta) = I(\theta).
    \)
    
    Therefore $I(\theta)$ is the fixed point of the linear action of the group $\mathcal R_\theta$.
    However, as $\mathcal R_\theta$ is acting as the full rotation group on the orthogonal complement of $\theta$, the only fixed point in $\theta^{\perp}$ is the zero vector, and the only dimension left is for the span of the vector itself; hence $I(\theta)=c_d\cdot\theta + 0$ for some scalar $c_d$.
    To calculate this constant, we set $\theta=e_1$ to obtain
    \begin{align*}
        c_d
        &\defeq
         \int_{\sdm}
             \1\{q_1 \geq 0\}q_1 \: d\bar{\sigma}(q)
        =
        \frac{1}{2} \frac{\Gamma(d/2)}{\sqrt{\pi}\Gamma((d-1)/2)} \int_{0}^{1}
            t (1-t^2)^{\frac{d-3}{2}} dt \notag  \\
        &=
        \frac{1}{d-1} \frac{\Gamma(d/2)}{\sqrt{\pi}\Gamma(\frac{(d-1)}{2})} 
        = 
        \frac{\Gamma(d/2)}{2\sqrt{\pi}\Gamma(\frac{d+1}{2})} 
        \geq \frac{1}{\sqrt{2\pi}} \cdot d^{-1/2},
        \label{eq:cd-volume-constant}
    \end{align*}
    where the second to last step follows from the Gamma recurrence $\Gamma(x+1)=x\Gamma(x)$, and the last  using Wendel's Inequality~\cite{wendel1948note,luo2012bounds}; plugging in $x = d/2$ and $s = 1/2$ yields $\frac{\Gamma(d/2)}{\Gamma(d/2 + 1/2)} > \frac{\sqrt{2}}{d^{1/2}}$.
\end{proof}
    
This identity allows us to recover the first moment.
\begin{lemma}
\label{lem:first-moment-identification}
    The first moment $M_1$ is identifiable with access to $\Q_1$.
\end{lemma}
\begin{proof}
    By applying first \Cref{lem:first_identity} and then Fubini's theorem, we can write $M_1$ as
    \begin{align*}
        \E{\theta \sim \Theta}{\theta}
        &= \E{\theta \sim \Theta}{\frac{1}{c_d} \cdot \int \1\{\ip{\theta}{q}\ge 0\} \cdot  q \: d\bar{\sigma}(q) }
        = \frac{1}{c_d} \cdot \int \int \1\{\ip{\theta}{q}\ge 0\} \cdot q \: d\bar{\sigma}(q) \:d\Theta(\theta) \\
        &= \frac{1}{c_d} \cdot \int \int \1\{\ip{\theta}{q}\ge 0\} \cdot  q \: d\Theta(\theta) \: d\bar{\sigma} (q)
        = \frac{1}{c_d} \cdot \int \Q_1(q) \cdot q \: d\bar{\sigma} (q). \qedhere
    \end{align*}
\end{proof}

Combining \Cref{lem:first-moment-identification} with 
Observation~\ref{obs:welfare-max-from-first-moment}, we have our first major result.
\begin{theorem}
\label{thm:welfare-max-from-pairwise-queries}
    Welfare-maximizing candidates are identifiable from 
    $\Q_1$.
\end{theorem}

\subsection{Pairwise Queries and Higher Moments}
\label{sec:2-moment}

We may naturally wonder whether $\Q_1$ suffices to identify higher moments.
As we show in the following example, it does not even for the second moment.
\begin{example}
\label{ex:antipodal-voter-distribution}
    Consider the voter type distribution $\Theta_{\pm \theta}$ over $\sdm$ given by placing half of the probability mass in an $\eps$-neighborhood around the vector $\theta$ and the other half (antipodally) symmetrically around $-\theta$. 
    By symmetry, $Q_1(q) = 1/2$ for all queries $q$ regardless of $\theta$.
    Furthermore,
    all candidates $\phi \in \sdm$ confer expected welfare $0$.
    But candidates orthogonal to $\theta$ have variance $\approx 0$, while candidates $c \defeq \beta \theta$ for a fixed $\beta$ have variance $\beta^2$.
\end{example}

This example %
illustrates that some distributions are indistinguishable from one another via $\Q_1$ and therefore from their first moments $M_1$, but that optimizing well-motivated nonlinear objectives requires distinguishing them.
Naturally, this trend continues: knowing the moment tensors $M_1, \ldots, M_k$ does not determine subsequent moments, or the distribution overall, even when its support is constrained to $\sdm$. 
\begin{observation}
\label{thm:k-moment-gap}
    For $d \ge 2$ and let $k \ge 1$.
    There exist two probability measures $\mu_+$ and $\mu_-$ on the sphere $\sdm$ such that they have the same first $k$ moments but have different $(k+1)$-st moments.
\end{observation}
We defer details to \Cref{app:proofs}. 
In particular, this means that we should not hope to derive $M_2$ from $M_1$. 
Can we derive $M_2$ from $\Q_2$?
We now show that we can; in fact, we show that we can derive $M_k$ from $\Q_k$, for any $k$.

\begin{theorem}
\label{thm:kth-moment-identification-from-k}
    The $k$th moment tensor $M_k$ is identifiable from
    \begin{equation*}
        M_k = \frac{1}{c_d^{k}} \cdot 
        \E{q_1,\dots,q_k\sim\bar{\sigma}^k}{\Q_k(q_1,\dots,q_k)\cdot  q_1 \otimes \cdots \otimes q_k},
    \end{equation*}
    where  $c_d$ is defined as in \Cref{lem:first_identity}.
\end{theorem}

\begin{proof}
    In Lemma~\ref{lem:first_identity}, we showed that for any fixed $\theta$ and for uniform $q\sim\bar{\sigma}$,
    \(
        \E{q}{\1\{ \langle \theta, q \rangle\ge 0 \} \cdot q} = c_d \cdot \theta.
    \)
    We will now generalize this to the expectation over $k$ independently chosen pairwise comparison directions.
    Consider sampling $k$ uniformly independent comparison directions $\bq \sim \bar{\sigma}^k$, where $\bq = (q_1, \ldots, q_k)$.
    For fixed $\theta$, due to independence, 
    \begin{align*}
        \E{\bq}{\1\{\langle\theta,q_1\rangle\ge0,\dots,\langle\theta,q_k\rangle\ge0\} \cdot q_1\otimes\cdots\otimes q_k} %
        &= \bigotimes_{j=1}^k \E{q_j}{\1\{\langle\theta, q_j\rangle \ge 0 \} \cdot q_j }
        = c_d^{k} \cdot \theta^{\otimes k},
    \end{align*}
    where last step follows from \Cref{lem:first_identity}.
    Taking the expectation over $\Theta$ and applying Fubini yields
    \begin{align*}
        \E{\theta \sim \Theta}{\E{\bq}{\1 \{ \langle\theta,q_1\rangle \ge 0, \dots, \langle\theta,q_k\rangle \ge 0 \} \cdot q_1\otimes\cdots\otimes q_k } } 
        &=\E{\bq}{Q_k(\bq)\cdot q_1\otimes\cdots\otimes q_k} \\
        &= c_d^{k} \cdot \E{\theta \sim \Theta}{\theta^{\otimes k}}.
    \end{align*}
    Rearranging gives the stated claim.
\end{proof}

\subsection{Identifying the Voter Distribution via Size-2 Queries}

Having established that $k$th moments of $\Theta$ are identifiable from $k$ pairwise comparison queries,
we now show an even stronger result: \emph{the full voter distribution is identifiable from two pairwise queries}.
Although in Section~\ref{sec:estimation}, we observe that the associated sample complexity is higher, it is quite remarkable that this is even possible.

First, we note that 
the set of all moments $\Moms$ uniquely identify $\Theta$.
This is a well-known fact in probability and real analysis, even for the more general setting where the support of the distribution is a compact subset of $\R^d$ (e.g. \citep[Corollary 14.9]{schmudgen17moment}).
\begin{lemma}
\label{lem:uniqueness-of-Theta-from-moments}
    If voter distributions $\Theta$ and $\Theta'$ have equal moments $\Moms(\Theta) = \Moms(\Theta')$, then they are in fact equal: for all measurable sets $T\subseteq \sdm$ it holds that $\Theta(T) = \Theta'(T)$.
\end{lemma}
We outline the argument for completeness in Appendix~\ref{app:proofs}.
Combined with \Cref{thm:kth-moment-identification-from-k}, this says that if we can ask each voter infinitely many queries, we can distinguish any two distributions $\Theta$ and $\Theta'$. The problem is, of course, that asking a single voter an unbounded number of queries is infeasible.
We now demonstrate that this is unnecessary.%

\begin{theorem}
\label{thm:identifiability-from-two}
    Distributions are identifiable from $2$-sized queries.
\end{theorem}

The high-level proof relies on the spectral decomposition of functions on the sphere into \emph{spherical harmonics}. Recall that a single query effectively measures the fraction of voters residing in a specific hemisphere. In the spherical harmonic literature, this mapping is known as the \emph{hemispherical transform}. This transform is invertible on the subspace of odd-degree spherical harmonics, allowing us to identify $\mathbb{E}_{\Theta}[p(\theta)]$ for any odd-degree harmonic $p$. Using linearity of expectation, this allows us to reconstruct all odd-degree moments.

To identify even moments, we observe that any even-degree polynomial can be decomposed into a sum of a product of two odd-degree harmonics. This reduces the problem to identifying values of the form $\mathbb{E}_{\Theta}[p(\theta)p'(\theta)]$. By asking two queries to each arriving voter, we can estimate the joint expectation of the responses, which identifies this product and allows us to recover all even moments.

\begin{proof}[Proof of \Cref{thm:identifiability-from-two}]
We will show that all moments are identifiable.  \Cref{lem:uniqueness-of-Theta-from-moments} then implies that distributions are identifiable.

Our analysis relies on the spectral theory of functions on the sphere~\citep[Chapter 3]{groemer1996geometric}. 
For each $j \ge 0$, let $\mathcal{H}_j$ denote the finite-dimensional space of \emph{spherical harmonics} of degree $j$ on $\sdm$. An important fact is that any homogeneous polynomial can be decomposed into a sum of spherical harmonics~\cite[Lemma 3.2.5]{groemer1996geometric}. Specifically, if $p$ is a homogeneous polynomial of degree $k$, we can write it as a sum of spherical harmonics, one for each degree $j \le k$ such that $k-j$ is even. We begin with odd $k$, and show how to extend our analysis to even $k$ later. In particular, if $k$ is odd, there exists $f_j \in \mathcal{H}_j$ for each odd $j \le k$ such that for all $\theta \in \sdm$,
\begin{equation*}
\label{eq:polynomials-decompose-into-harmonics}
    p(\theta) = \sum_{\substack{j \le k \ j \text{ is odd}}} f_j(\theta) \:.
    \footnote{Technically, \citet{groemer1996geometric} proves this for polynomials \emph{not} restricted to the sphere for all $x \in \mathbb{R}^d$, and it is of the form $p(x) = \sum_{j \le k: j \text{ is odd}} \|x\|^{k - j}f_j(x)$. However, restricted to the sphere, $\|x\| = 1$.}
\end{equation*}

For our purposes, $p$ will be a $k$-degree monomial (e.g., $p(\theta) = \prod_{j = 1}^k \theta_{i_j})$ for some indices $i_1, \ldots, i_k$). Our goal is equivalent to recover $\E{\theta \sim \Theta}{p(\theta)}$ for all monomial $p$.

We now use a second property of spherical harmonics, i.e., Funk-Hecke Theorem (see Thm.~\ref{thm:funk-hecke}). It suggests that for any fixed spherical harmonic $f \in \mathcal{H}_j$ with $j$ being odd, we have
\[
    f(\theta) = \frac{1}{\mu_j}\int_{S^{d - 1}} \1\{\langle \theta, q \rangle \ge 0\} f(q) \; d\bar{\sigma}(q),
\] for a nonzero constant $\mu_j$.

Taking the expectation over $\Theta$, 
\begin{align*}
\E{\theta \sim \Theta}{f(\theta)} &= \frac{1}{\mu_j} \iint \1\{\langle \theta, q \rangle \ge 0\} f(q) \; d\bar{\sigma}(q)d\Theta(\theta)
    =  \frac{1}{\mu_j} \iint \1\{\langle \theta, q \rangle \ge 0\} f(q) \;d\Theta(\theta) d\bar{\sigma}(q)\\
    &= \frac{1}{\mu_j} \int f(q)\int \1\{\langle \theta, q \rangle \ge 0\} \;d\Theta(\theta) d\bar{\sigma}(q).
\end{align*}
The exchange of integrals is justified by Fubini's theorem, as the integrand is bounded ($\1\{\langle \theta, q \rangle \ge 0\}$ is trivially bounded, and spherical harmonics are continuous functions on the compact sphere), and the measures are finite. The inner integral corresponds to the expected query response: \[
    \int \1\{\langle \theta, q \rangle \ge 0\} \,d\Theta(\theta) = \Pr_\Theta[\langle q, \theta \rangle \ge 0] = Q_1(q). 
\]
Thus, the expected value of the harmonic is fully determined by $Q_1$:
\[
    \E{\theta \sim \Theta}{f(\theta)} =  \int f(\theta) \: d\Theta(\theta) =  \frac{1}{\mu_j}  \int f(q) Q_1(q, 1)\: d\bar{\sigma}(q).
\]

Next, consider even $k$. We can decompose any monomial $p$ of degree $k$ into a product of two monomials of odd degrees $p_1$, $p_2$ of degrees $k_1$ and $k_2$ (e.g., by partitioning the index set such that $k_1 + k_2 = k$). As we have just seen, both $p_1$ and $p_2$ have decompositions into odd-degree spherical harmonics $f_j \in \mathcal{H}_j$ for odd $j \le k_1$ and $f'_{j'} \in \mathcal{H}_{j'}$ for $j' \le k_2$ such that,
\[
    p_1 = \sum_{j \le k_1: j \text{ is odd }} f_j \qquad \text{and} \qquad p_2 = \sum_{j' \le k_2: j' \text{ is odd }} f'_{j'}.
\]
Thus, we can write
\begin{equation}
\label{eq:even-polynomials-decompose-into-harmonics}
    p(x) = \sum_{j \le k_1, j' \le k_2, j, j' \text{ are odd }} f_j(x) \cdot f'_{j'}(x).
\end{equation}
By \eqref{eq:even-polynomials-decompose-into-harmonics}, to determine $\E{\theta \sim \Theta}{p(\theta)}$ for even-degree $p$ it therefore suffices to deduce the expectation of the product of any two odd-degree spherical harmonics $f$ and $f'$; that is, $\E{\theta \sim \Theta}{f(\theta)f'(\theta)}$.

Fix two such spherical harmonics $f$ and $f'$ of odd degrees $j$ and $j'$. As above, we can write
\[
    f(\theta) = \frac{1}{\mu_j} \int \1\{\langle \theta, q \rangle \ge 0\} f(q) \: d\bar{\sigma}(q)
    \qquad \text{and} \qquad 
    f'(\theta) = \frac{1}{\mu_{j'}} \int \1\{\langle \theta, q \rangle \ge 0\} f'(q) \: d\bar{\sigma}(q).
\]
By the linearity of the integral, 
\[
    f(\theta) \cdot f'(\theta) = \frac{1}{\mu_j \mu_{j'}} \iint \1\{\langle \theta, q \rangle \ge 0\}\1\{\langle \theta, q' \rangle \ge 0\}f(q)f'(q') \: d\bar{\sigma}(q) \: d\bar{\sigma}(q').
\]
Taking the expectation with respect to $\Theta$, we have
\begin{align*}
    \E{\theta \sim \Theta}{f(\theta) \cdot f'(\theta)}
    &= \frac{1}{\mu_j \mu_{j'}} \iint \1\{\langle \theta, q \rangle \ge 0\}\1\{\langle \theta, q' \rangle \ge 0\}f(q)f'(q') \: d\bar{\sigma}(q) \: d\bar{\sigma}(q') \\
    &= \frac{1}{\mu_j \mu_{j'}} \iint f(q)f'(q') \left[\int \1\{\langle \theta, q \rangle \ge 0\}\1\{\langle \theta, q' \rangle \ge 0\}\, d\Theta(\theta) \right] \: d\bar{\sigma}(q) \:d\bar{\sigma}(q'),
\end{align*}
where we again use Fubini's theorem for the integral swap. Finally, letting $\bq = (q, q')$ we have
\begin{align*}
    \int \1\{\langle \theta, q \rangle \ge 0\}\1\{\langle \theta, q' \rangle \ge 0\} \: d\Theta(\theta)
    = \Pr_\Theta[\1\{\langle q, \theta \rangle \rangle \ge 0 \} \land \1\{\langle q', \theta \rangle \rangle \ge 0 \}] 
    = Q_2(\bq).
\end{align*}
Hence,
\[
    \E{\theta \sim \Theta}{f(\theta) \cdot f'(\theta)} = \frac{1}{\mu_j \mu_{j'}}\iint f(q)f'(q') Q_2(\bq)  \: d\bar{\sigma}(q) \: d\bar{\sigma}(q'),
\]
implying identifiability.
\end{proof}

\subsection{Identifying the Voter Distribution via Graded Queries}

Next, we consider graded queries. While these require the stronger Assumption~\ref{assum:express_2}, they yield a powerful return: the entire distribution is identifiable using only a single query per voter.

\begin{theorem}
\label{thm:moment_identity_graded}
For almost all $\tau \in (0,1)$, distributions are identifiable with access to $G_{\tau}$.  
\end{theorem}

The proof is deferred to \Cref{app:proofs}. The high level approach is to show that each moment $k$ is identifiable by (strong) induction on $k$. Suppose we know the first $k - 1$ moments. We consider the correlation between the graded response probability $G_\tau(q)$ and the $k$-th tensor power of the query vector $q$, integrated over the uniform distribution of queries. Using the rotational symmetry of the sphere, this integral simplifies into a linear equation involving the unknown $k$-th moment of the voter distribution scaled by a specific scalar coefficient, plus terms composed entirely of lower-order moments (which are known by the inductive hypothesis). Crucially, this scaling coefficient depends on the threshold $\tau$ and behaves like a polynomial with a finite number of roots. Therefore, for almost all choices of $\tau$ (i.e., any $\tau$ not in this finite set of roots), the coefficient is non-zero, allowing us to invert the equation and uniquely recover the $k$-th moment. By excluding the countable union of such roots across all $k$, we can identify all moments, and hence the entire distribution, for almost every $\tau$.
Notably, if we only care about the first two moments, any $\tau \in (0,1)$ suffices.

\section{Moment Estimation}
\label{sec:estimation}

In order to make our identifiability results useful, we must contend with the fact that we do not have perfect knowledge of these distributions, but instead can only obtain a finite set of samples of such comparisons. 
How many voters are sufficient in order to approximately learn the moments of $\Theta$? In general, we could have an algorithm that chooses which query $\mathbf{q}$ to make based on all previous responses. For all of our positive results, it will suffice to have $\mathbf{q}$ be selected uniformly at random.
Our goal in this section is to bound the \emph{sample complexity} of moment estimation.
\begin{definition}[Sample Complexity]
    Estimating the $k$th moment from comparison queries has \emph{sample complexity} $T$ if there is an estimator $\widehat{M}_k$ taking in $T$ $t$-sized regular query-response pairs $\bm{r}=\{\bq_i, (\resp_{\theta_i}(\bq_i)\}_{i=1}^{T}$ of $T$ i.i.d. voters $\theta_i \sim \Theta$ to $T$ uniformly random $t$-sized queries $\bq_1, \ldots, \bq_T \sim \bar{\sigma}^t$,
    such that for all unknown $\Theta$,
    \[
        \Pr_{\bm{r}}\left[\norm{\widehat{M}_k(\bm{r}) - M_k(\Theta)} \le \varepsilon \right] \geq 1-\delta.
    \]
\end{definition}

\subsection{Warmup: Welfare Maximization from Queries}

Our first result shows that we can successfully estimate the first moment of the voter distribution by requesting a response to a single pairwise query from a polynomial number of voters.
\begin{theorem}\label{thm:complexity-first}
    The sample complexity of estimating the first moment from single queries is $O(\frac{d}{\eps^2}\log\frac{1}{\delta})$.
\end{theorem}

\begin{proof}
	Consider the double integral over the probability measure for both the voter distribution $\Theta$ and the query distribution $\Unif{\sph^{d-1}}$, which we denote by $d\Theta$ and $d\bar{\sigma}(q)$ respectively:
	\begin{align*}
		\iint_{\sdm \times \sdm}
		\resp_\theta(q) \cdot q \: d\bar{\sigma}(q) \: d\Theta(\theta)
		=\int_{\sdm} c_d\cdot \theta  \: d\Theta(\theta) =c_d \cdot \bar \theta. \label{eq:double}
	\end{align*} 
	The corresponding estimator from $T$ samples is
	$\widehat{\bar{\theta}}=\frac{1}{c_d T}\cdot \sum_{i \in [T]}\1\left\{\langle q_i, \theta_i \rangle \ge 0\right\} \cdot q_i. $ 
	
	We would like to bound the rate at which this approaches $\bar \theta$ in spectral norm, which is just $L_2$ norm. That is, we would like to upper bound the $T$ required to satisfy   \[
        \Pr\left[ \norm{\bar \theta - \widehat{\bar\theta} }_2 > \eps \right] \leq \delta.
    \]
   
    We take a standard approach and appeal to McDiarmid's inequality~\cite{mcdiarmid89method}.
    To this end, let $\hat \theta_i \defeq \frac{1}{c_d} \cdot \resp_{\theta_i}(q_i) \cdot q_i$ define the auxiliary random variable $y_i \defeq \bar \theta - \hat \theta_i$, and let $f(y_1, \ldots, y_T) \defeq \norm{\sum_i y_i}_2 = \norm{\bar \theta - \widehat{\bar\theta} }_2.$
    Our goal is then equivalent to upper bounding $\Pr\left[ f(y_1, \ldots, y_T) \geq \eps \cdot T\right]$.
    By the triangle inequality, this $f$ satisfies the bounded difference property that for all $y_i$ and $y_i'$ it holds that $\abs{f(y_1, \ldots, y_i, \ldots, y_T) - f(y_1, \ldots, y_i', \ldots, y_T)} \leq \norm{\hat \theta_i - \hat \theta_i'}_2 = O(c_d^{-1})$, since our sampled $q_i$ have unit norm.
    Therefore by McDiarmid's inequality, 
    \[
        \Pr\left[ \norm{\bar \theta - \widehat{\bar\theta} }_2 > \eps \right] = \Pr\left[ f(y_1, \ldots, y_T) \geq \eps \cdot T\right] \leq \exp \left( \frac{-2\eps^2 T }{C \cdot c_d^{-2}}\right)
    \]
    for some constant $C$. 
    Setting this upper bound on the failure probability to $\delta$, solving for $T$, and recalling that $c_d = \Theta(d^{-1/2})$ finally yields $T = O(\frac{d}{\eps^2} \log \frac{1}{\delta})$, as claimed. 
\end{proof}

\subsection{Moments from $k$-Wise Queries}
Our next result significantly generalizes our observation above that the first moment can be efficiently estimated, showing that this, in fact, holds true more generally.
\begin{theorem}
\label{thm:mk_matrix_bernstein}
    The sample complexity of estimating the $k$th moment from $k$-sized queries is 
    \[ O\left(\frac{(2\pi)^k \cdot k \cdot d^{\lceil k/2\rceil}}{\eps^2}
    \log\left(\frac{d}{\delta}\right)\right). \]
\end{theorem}

A technical challenge is that tensors are not simply matrices with more indices; even extending familiar matrix notions—such as computing the norm of higher-order tensors is nontrivial~\cite{hillar2013most}. And there is no comparably sharp, general-purpose concentration theory for tensors. We therefore work with matricizations of the relevant moment tensors, which allows us to leverage the matrix Bernstein inequality~\cite[Theorem~1.6.2]{tropp2012user}.
\begin{theorem}[Matrix Bernstein Inequality for real-valued matrices]
    Let $S_{1}, \dots, S_{n}$ be independent, centered real random matrices with common dimension $d_{1} \times d_{2}$, and assume that each one is uniformly bounded, i.e., 
    $\E{}{S_{k}} = 0$ and $\| S_{k} \| \le L$ for each $k = 1, \dots, n$.
    Let $Z = \sum_{k=1}^{n} S_{k}$, and let $\nu(Z)$ denote the matrix variance statistic of the sum:
    \[
        \nu(Z) = \max\left\{ \|\mathbb{E}(ZZ^T)\|, \|\mathbb{E}(Z^TZ)\| \right\} = \max\left\{ \left\| \sum_{k=1}^{n} \mathbb{E}(S_{k}S_{k}^T) \right\|, \left\| \sum_{k=1}^{n} \mathbb{E}(S_{k}^TS_{k}) \right\| \right\}.
    \]
    Then, for all $t \ge 0$,
    \[
        \mathbb{P}\left\{ \|Z\| \ge t \right\} \le (d_{1} + d_{2}) \cdot \exp\left( \frac{-t^{2}/2}{\nu(Z) + Lt/3} \right).
    \]
\end{theorem}

The following standard result shows that it is sufficient to bound the norm of the matricization.
\begin{lemma}[Proposition 4.1, \cite{wang2017operator}]
\label{lem:inj_le_unfold}
Let $T\in(\mathbb R^d)^{\otimes k}$ and let $I\sqcup J=\{1,\dots,k\}$.
Then
\(
    \|T\|  \le  \|\Mat_{I|J}(T)\|.
\)
\end{lemma}
Note that had we not matricized, typical vector-based concentration inequalities would end up with a dominant $d^k$ factor. The matricization and more sophisticated matrix Bernstein inequality allow us to cut  $\approx d^{\lfloor k/2 \rfloor}$ off of the sample complexity.

\begin{proof}[Proof of \Cref{thm:mk_matrix_bernstein}]
We will make use of \Cref{thm:kth-moment-identification-from-k}. Recall that $\Q_k(q_1, \ldots, q_k)$ is the probability that a random voter $\theta$ has $\resp_\theta(q_i) = 1$ for all $i$. Thus, we can rewrite the statement as
\[
    M_k = \frac{1}{c_d^k} \cdot \E{q_1, \ldots, q_k \sim \bar{\sigma}^k, \theta \sim \Theta}{\1\{\resp_\theta(q_i) = 1 \, \forall i \} \cdot q_1 \otimes \cdots \otimes q_k}.
\]
Given $T$ sampled queries and responses $\bm{r}=\{\bq_i, (\resp_{\theta_i}(\bq_i)\}_{i=1}^{T}$, we let $\chi_i \defeq \1\{\resp_{\theta_i}(q_{i, j}) = 1 \, \forall j \}$ be the indicator that the $i$-th voter agreed with all $k$ queries.
We define the empirical estimator for the $k$th moment as:
\[
    \widehat M_k(\bm{r}) \defeq \frac{1}{T} \sum_{i=1}^T \frac{\chi_i}{c_d^{\,k}} (q_{i,1}\otimes\cdots\otimes q_{i,k}),
\]
which, by above, is an unbiased estimate of $M_k$.

To bound the estimation error $\norm{\widehat M_k - M_k}$, we fix an arbitrary partition of the modes $I\sqcup J=\{1,\dots,k\}$ with $|I|=s$ and $|J|=k-s$. Let $(d_1,d_2)\defeq(d^s,d^{k-s})$. By Lemma~\ref{lem:inj_le_unfold}, it suffices to bound the spectral norm of the matricized difference as
\(
    \norm{\Mat_{I|J}(\widehat M_k - M_k)}.
\)
We define independent random matrices $A_i \in \R^{d_1\times d_2}$ such that:
\[
    A_i \defeq \Mat_{I|J}\left(\frac{1}{c_d^{k}} \cdot \chi_i \cdot q_{i,1}\otimes\cdots\otimes q_{i,k}\right).
\]
Notice that $\E{}{A_i} = \Mat_{I|J}(M_k).$ Then, $\Mat_{I|J}(\widehat M_k - M_k)=\frac1T\sum_{i=1}^T S_i$, where $S_i \defeq A_i-\E{}{A_i}$ are independent, zero-mean random matrices. To apply the Matrix Bernstein inequality,  we must compute a uniform bound $L$ on the spectral norm of the summands and a bound $\nu$ on the total variance statistic.

Towards establishing this uniform bound, for each $i$ we can write the random matrix as an outer product $A_i = \frac{\chi_i}{c_d^{\,k}} \cdot u_i v_i\tp$, where $u_i \defeq \bigotimes_{j\in I} q_{i,j}\in\R^{d_1}$ and $v_i \defeq \bigotimes_{j\in J} q_{i,j}\in\R^{d_2}$. Using the identity that $\|x \otimes y\| = \|x\| \cdot \|y\|$, since the queries are all unit norm, $\norm{u_i}_2=\norm{v_i}_2=1$. Because $\chi_i \in \{0, 1\}$, we have:
\(
    \norm{A_i} = \frac{\chi_i}{c_d^{\,k}} \cdot \norm{u_i}_2 \norm{v_i}_2 \le \frac{1}{c_d^{\,k}}.
\)
Since all norms are convex, by Jensen's inequality $\norm{\E{}{A_i}}\le \E{}{\norm{A_i}}\le \max_i \norm{A_i} \le c_d^{-k}$. Then by the triangle inequality, we have
\[
    \norm{S_i} \le \norm{A_i}+\norm{\E{}{A_i}} \le \frac{2}{c_d^{\,k}} =: L.
\]

We now bound the variance statistic $\nu \defeq \max\left\{\norm{\sum_{i=1}^T \E{}{S_i S_i\tp}}, \norm{\sum_{i=1}^T \E{}{S_i\tp S_i}}\right\}$. 
Since the $T$ voters are  independent, the total variance is bounded by:
\begin{equation}
\label{eq:matrix-variance-independent-voters-bound}
    \nu \le T\cdot \max_i \max\left\{\E{}{\norm{S_i S_i\tp}}, \norm{\E{}{S_i\tp S_i}} \right\}.
\end{equation}
Next note that $\E{}{S_i S_i\tp} = \E{}{A_iA_i\tp} - \E{}{A_i}\E{}{A_i}\tp$. We show that this implies 
\begin{equation}
\label{eq:variance-S-vs-A-norm-bound}
    \norm{\E{}{S_i S_i\tp}} \le \norm{\E{}{A_iA_i\tp}}.
\end{equation}
Indeed, for any unit vector $x$, 
\begin{align*}
    x\tp\E{}{S_i S_i\tp}x
    = x^\top \left( \mathbb{E}[A_i A_i^\top] - \mathbb{E}[A_i]\mathbb{E}[A_i]^\top \right) x
    = x^\top\mathbb{E}[A_i A_i^\top]x - x^\top \mathbb{E}[A_i]\mathbb{E}[A_i]^\top x.
\end{align*}
Letting $y = \mathbb{E}[A_i]^\top x$, this is equal to
\(
     x^\top\mathbb{E}[A_i A_i^\top]x - y\tp y.
\)
Since $y\tp y \ge 0$, this implies $x^\top\mathbb{E}[A_i A_i^\top]x \ge x\tp\E{}{S_i S_i\tp}x$ for all $x$. Hence $\norm{\E{}{S_i S_i\tp}} \le \norm{\mathbb{E}[A_i A_i^\top]}$.

With this established, we turn to upper bounding $\norm{\E{}{A_iA_i\tp}}$.
Utilizing the structure of $A_i$, the fact that $v_i$ is a unit vector (and hence $v_i\tp v_i = 1$) and $\chi_i^2 = \chi_i \le 1$, we have:
\[
    \norm{A_iA_i\tp} = \norm{\frac{\chi_i^2}{c_d^{2k}} \cdot u_i (v_i\tp v_i) u_i\tp } \le  \norm{ \frac{1}{c_d^{2k}} \cdot u_i u_i\tp}.
\]
Since the query vectors $q_{i,j}$ are drawn independently and uniformly from $\sdm$, $\E{}{q_{i,j} q_{i,j}\tp} = \frac{1}{d}I_d$. By independence across the distinct queries:
\[
    \E{}{u_i u_i\tp}
    =
    \big(\E{}{q_{i,j} q_{i,j}\tp}\big)^{\otimes s}
    =
    \left(\frac{1}{d}I_d\right)^{\otimes s}
    =
    \frac{1}{d^s}I_{d_1}.
\]
Furthermore, $\norm{\frac{1}{d^s}I_{d_1}} = \frac{1}{d^s}$, so we get that
\(
    \norm{\E{}{A_iA_i\tp}} \le \frac{1}{c_d^{2k}d^s}.
\)
Applying the analogous argument to $\norm{\E{}{A_i\tp A_i}}$, we get
\(
    \norm{\E{}{A_i\tp A_i}} \le \frac{1}{c_d^{2k}d^{k-s}}.
\)
From \eqref{eq:matrix-variance-independent-voters-bound} and \eqref{eq:variance-S-vs-A-norm-bound}, this implies $\nu \leq \frac{T}{c_d^{2k}} \max \left(d^{-s},\: d^{s - k} \right) =  \frac{T}{c_d^{2k}} d^{-\min(s, k - s)}$.
Applying the Matrix Bernstein inequality then yields:
\[
    \mathbb \Pr\left(\norm{\frac{1}{T}\sum_{i=1}^T S_i} \ge \eps\right)
    \le (d_1+d_2)\exp\left(-\frac{T^2\eps^2/2}{\nu + L T\eps/3}\right).
\]
To ensure the error satisfies $\norm{\widehat M_k - M_k} \le \eps$ with probability at least $1-\delta$, we enforce the failure probability bound $\le \delta$. Substituting the derived bounds for $\nu$ and $L$, a sufficient condition for the sample complexity is:
\[
    T \ge
    \frac{2}{\eps^2}\left(\frac{1}{c_d^{2k}d^{\min(s,k-s)}}+\frac{2\eps}{3c_d^{k}}\right)
    \log\left(\frac{d^s+d^{k-s}}{\delta}\right).
\]
We set $s=\lfloor k/2\rfloor$. Furthermore, recall from \Cref{lem:first_identity} that $c_d > \frac{1}{\sqrt{2\pi}} d^{-1/2}$. Furthermore, it is without loss of generality to assume $\varepsilon \le 1$, as the $0$-tensor is a trivial $1$-approximation as $\norm{M_k} \leq 1$. 
Thus,
\begin{align*}
    T &= O\left(\frac{1}{\varepsilon^2}\left(\frac{d^{- \lfloor k/2 \rfloor}}{c_d^{2k}} + \frac{1}{c^k_d}{}\right)
    \log\left(\frac{d^{\lceil k/2\rceil}}{\delta}\right)\right)
    = O\left(\frac{1}{\varepsilon^2}\left(d^{k- \lfloor k/2 \rfloor}\cdot (2\pi)^{k} + d^{k/2} \cdot (2\pi)^{k/2}\right)
    \log\left(\frac{d^{\lceil k/2\rceil}}{\delta}\right)\right)\\
    &= O\left(\frac{(2\pi)^k \cdot k \cdot d^{\lceil k/2\rceil}}{\eps^2}
    \log\left(\frac{d}{\delta}\right)\right). \qedhere
\end{align*}
\end{proof}

\subsection{Estimating All Moments from Few Queries}

Our general result about moment estimation above requires $k$ queries per voter.
In Section~\ref{sec:identifiability}, on the other hand, we showed that we can \emph{identify} all moments with only two regular queries.
How does this translate into sample efficiency?

The following result shows that it is still possible to estimate, although our upper bound is weaker, requiring on the order of $d^{3k + 2}$ rather than $d^{\lceil k/2 \rceil}$.
\begin{theorem} \label{thm:sample-size-2}
    The sample complexity of estimating the $k$th moment from size-$2$ queries is $$O\left( \frac{k (d + 1)^{3k + 2}}{\eps^2}\log\left(\frac{d}{\delta} \right)\right).$$ 
\end{theorem}
The proof is deferred to \Cref{app:proofs}. Our high-level approach mirrors the strategy used for $k$-wise queries in \Cref{thm:mk_matrix_bernstein}: we construct an unbiased estimator for the $k$-th moment tensor and bound its convergence using the Matrix Bernstein inequality. However, the construction of this estimator is considerably more involved than in the $k$-wise case, analogous to how \Cref{thm:identifiability-from-two} was more involved than \Cref{thm:kth-moment-identification-from-k}. Recall that the constructive proof of \Cref{thm:identifiability-from-two} establishes that the expectation of any degree-$k$ monomial can be recovered by integrating the product of two specific spherical harmonic functions against the query distribution $Q_2$. Consequently, we can define our estimator via the empirical mean of these harmonic functions evaluated on sampled query pairs. The central technical challenge lies in bounding the magnitude of this estimator; for this, we must turn to more explicit constructions of spherical harmonics. These magnitude bounds lead to less efficient reconstruction: the sample complexity scales as $O(d^{3k+2})$ compared to the $O(d^{\lceil k/2 \rceil})$ achieved with $k$-wise queries.

Regarding graded queries, we leave the sample complexity bound to future work. While \Cref{thm:moment_identity_graded} establishes identifiability for almost all thresholds $\tau$ and provides a corresponding unbiased estimator, the reconstruction relies on inverting a specific coefficient that depends on $\tau$. This coefficient vanishes at certain ``blind spots,'' and without a quantitative bound on how close an arbitrary $\tau$ is to these problematic points, we cannot control the magnitude or variance of the resulting estimator.

\section{Social Choice with Sparse Query Responses}
\label{sec:apps}

Classical social choice rules take full rankings as input. Here we show it is possible to design a complementary class of rules that only require extremely sparse query responses from voters. All missing proofs in this section are deferred to \Cref{app:proofs}.

\subsection{Moment-based Objectives}
We begin by considering our moment-based-objectives. Leveraging our sample complexity results for estimating moments, we can provide end-to-end guarantees. 

\smallskip
\noindent\textbf{Maximizing Social Welfare.} In a utilitarian setting, a natural and common objective is to maximize social welfare.
Our results imply that this can be done with polynomial sample complexity:
\begin{proposition}
\label{prop:welfare_max}
    By asking 1 query per voter, with at least $T \in O(\frac{d}{\eps^2} \log \frac{1}{\delta})$ arriving voters, with probability $1 - \delta$ over the arriving voters, for any candidate $\phi$ with $\|\phi\| \le B$, we can estimate the social welfare $\E{\theta}{u_\theta(\phi)}$ up to additive error $\varepsilon B$. In particular, given a context $x$ and $\ell$ candidates $y_1, \ldots, y_\ell$ such that $\|\Phi(x, y_i)\| \le B$, we can select $\hat{y}_i$ within $2\varepsilon B$ of the optimal welfare.
\end{proposition}

\smallskip
\noindent\textbf{Maximizing Risk-Adjusted Welfare.} A significant limitation of social welfare maximization is that it can yield extremely inequitable outcomes when voter preferences are highly diverse.
Risk-adjusted welfare addresses this by explicitly accounting for inequality aversion.
Armed with the first two moments, we can approximately maximize such welfare objectives.
\begin{proposition}
    \label{cor:risk_welfare}
    By asking 2 queries per voter, for all $\varepsilon \le 1$, with at least $T \in O( \frac{d}{\eps^2}
    \log \frac{d}{\delta})$ arriving voters, with probability $1 - \delta$ over the arriving voters, for any candidate $\phi$ with $\|\phi\| \le B$, we can estimate $\raw(\phi)$ up to $\sqrt{3} \cdot (\alpha + 1)B\sqrt{\varepsilon}$. In particular, given a prompt $x$ and $\ell$ possible responses $y_1, \ldots, y_\ell$ such that $\|\Phi(x, y_i)\| \le B$, we will be able to select one within $2\sqrt{3} \cdot (\alpha+1)B\sqrt{\varepsilon}$ of the optimal risk-adjusted welfare.
\end{proposition}

\subsection{Beyond Moment-based Objectives}
Next, we show how moments can be used to approximate other, more general objectives. 

\smallskip
\noindent\textbf{Maximizing Nash Welfare.} Like risk-adjusted welfare, Nash welfare,  aims to balance social welfare and distributed fairness. Defined as $\text{Nash}(\phi) = \E{\Theta}{\log(u_\theta(\phi))}$, this objective corresponds to maximizing the geometric mean of utilities (or the product in the finite case). Naturally, this is only well-defined if voter utilities are strictly positive. While Nash welfare cannot be computed exactly from a finite set of moments, we show that if the candidate's induced utilities lie within a known bounded interval, we can efficiently approximate the objective using moments.
\begin{theorem}
\label{thm:nash_approx}
Let $\phi$ be a candidate and let $[a, b] \subset (0, \infty)$ be a known interval such that voter utilities satisfy $u_{\theta}(\phi) \in [a, b]$ almost surely. Let $r \defeq a/b$. The Nash welfare can be estimated using the first $k$ moment tensors $M_1, \dots, M_k$ up to an additive error of:
\[
    \left| \text{Est}_k - \text{Nash}(\phi) \right| \le \frac{\sqrt{r}-1}{k+1} \left( 1 - \frac{2}{\sqrt{r}+1} \right)^{k}.
\]
\end{theorem}

Thus, we can achieve an approximation that improves exponentially quickly once $k \in \Omega(\sqrt{r})$. The proof uses standard techniques for approximating logarithm using $k$th-degree polynomials, specifically, using \emph{Chebyshev polynomials}~\cite{mason2002chebyshev}. The expectation of any $k$th-degree polynomial can then be computed using the first $k$ moments.

\begin{proof}[Proof of \Cref{thm:nash_approx}] The proof proceeds in two steps: first, we construct a $k$th-degree polynomial $P_k$ such that 
\[
    |P_k(x) - \log(x)|  \le \frac{\sqrt{r}-1}{k+1} \left( 1 - \frac{2}{\sqrt{r}+1} \right)^{k}
\]
for all $x \in [a, b]$. Then, we will show how to use $P_k$ to estimate Nash welfare. 

We will use \emph{Chebyshev polynomials} to approximate $\log$ on the interval. See \citet{mason2002chebyshev} for an overview. Each $T_k(x)$ is a $k$th-degree polynomial, defined on $[-1, 1]$, such that $|T_k(x)| \le 1$. They are defined such that $T_k(\cos(\theta)) = \cos(k \theta)$ (which, perhaps surprisingly, defines a $k$th-degree polynomial mapping $[-1, 1]$ to $[-1, 1]$).

We would like to approximate $\log(x)$ on $[a, b]$. We first reduce this to approximating on $[-1, 1]$.  We map the interval $[-1, 1]$ to $[a, b]$ using the affine transformation
\(
g(t) \defeq \frac{b+a}{2} + \frac{b-a}{2}t
\), and we let $\delta \defeq \frac{b-a}{b+a} = \frac{1-r}{1+r}$ where $r=a/b$.
Substituting this into the logarithm function, we have
\(
\log(g(t)) = \log(1 + \delta t) + \log((b + a)/2).
\)
Now, if we have a $k$th-degree polynomial $p'$ that is uniformly within $\varepsilon$ of $\log \circ g$ for $t \in [-1, 1]$, then $p'(g^{-1}(\cdot))$ is a $k$th-degree polynomial that is within $\varepsilon$ of $\log(x)$ on $[a, b]$. Furthermore, it suffices to find a $k$th-degree polynomial that approximates $\log(1 + \delta t)$, as $\log(g(t))$ differs from this function only by the constant $\log(\frac{b+a}{2})$.

To get this in a form using Chebyshev polynomials, we will make use of the identity~\cite[Equation 1.514]{gradshteyn2014table}
\[
    \log(1 - 2 \alpha \cos(\varphi) + \alpha^2) = -2\sum_{j = 1}^\infty \frac{\cos (j\varphi)}{j} \cdot \alpha^j,
\]
for $\alpha^2 \le 1$ and $\alpha \cos(\varphi) \ne 1$. In \Cref{app:proofs}, we show how to use this to derive that
\begin{equation}
\label{eq:log-as-chebyshev-sum}
    \log(1 + t\delta) = -\log(1 + \alpha^2) - 2\sum_{j = 1}^\infty \frac{\alpha^j}{j} \cdot T_j(t),
\end{equation}
for $\alpha = - \frac{\sqrt{r} - 1}{\sqrt{r} + 1}$. We can choose the $k$th-degree polynomial $R_k(t) \defeq -\log(1 + \alpha^2) - 2\sum_{j = 1}^k \frac{\alpha^j}{j} \cdot T_j(t)$ to approximate $\log(1 + \delta t)$. The error is
\begin{equation}
\label{eq:nash-chebyshev-approx-error-bound}
    \left|\log(1+t\delta) - R_k(t) \right| = \left| 2\sum_{j = k + 1}^\infty \frac{\alpha^j}{j} \cdot T_j(t) \right|.
\end{equation}
Using that $|T_j(t)| \le 1$ for $t \in [-1,1]$, we can upper bound this error by
\begin{align*}
    \sum_{j = k + 1}^\infty \frac{2|\alpha|^j}{j}
    &\le \frac{2}{k + 1} \sum_{j = k + 1}^\infty |\alpha|^j
    = \frac{2|\alpha|^{k + 1}}{(k+ 1)(1 - |\alpha|)}
    = \frac{\sqrt{r} + 1}{k + 1} \left(\frac{\sqrt{r} - 1}{\sqrt{r} + 1} \right)^{k + 1}\\
    &= \frac{\sqrt{r}  - 1}{k + 1} \left(\frac{\sqrt{r} - 1}{\sqrt{r} + 1} \right)^{k} = \frac{\sqrt{r}  - 1}{k + 1} \left(1 - \frac{2}{\sqrt{r} + 1} \right)^{k}.
\end{align*}

We have therefore constructed a $k$th-degree polynomial $P_k(x) \defeq R_{k} \circ g^{-1}  $, which can be written as $P_k(x) =  \sum_{j=0}^k c_j x^j$ for some coefficients $c_j$. Our proposed estimator is the expected value of this polynomial:$$\widehat{\text{Nash}} \defeq \mathbb{E}_{\theta \sim \Theta}[P_k(u_\theta(\phi))] =  \sum_{j=0}^k c_j \mathbb{E}_{\theta \sim \Theta}[(u_\theta(\phi))^j] = \sum_{j=0}^k c_j \cdot \langle M_j \cdot \phi^{\otimes j} \rangle .$$
Since $|P_{k+1}(u_\theta(\phi)) - \log(u_\theta(\phi))|$ is uniformly bounded for $u_\theta(\phi) \in [a,b]$ as in \eqref{eq:nash-chebyshev-approx-error-bound} by the arguments above, this yields the same bound on the Nash approximation.
\end{proof}

\smallskip
\noindent\textbf{Maximizing Welfare over Candidate Sets.} In the paradigm of pluralistic alignment, a natural objective is to identify a \emph{set} of candidates that represents most ``reasonable'' viewpoints~\cite{sorensen2024roadmap}.\footnote{This goes by the name of \emph{Overton} Pluralism.}
One way to operationalize this is to define the utility of each voter for a set of candidates as their \emph{maximum} utility over the candidates in this set.
In Section~\ref{sec:prelims}, we defined the associated welfare notion as \emph{top-choice welfare}.
Next, we show that we can obtain an approximate welfare maximizer in this setting by considering $k$ moments of the voter distribution.
\begin{theorem}[Approximability of $\ell$-$\tcw$ maximization from moments]
\label{thm:k-medians-approximable-from-moments}
    Fix $\eps > 0$, $\ell$, and a set of candidate responses $\Phi$. Let $\mathcal{W}_\ell \defeq \binom{\Phi}{\ell}$ be the collection of all sets of $\ell$ candidate.
    It is possible to identify a $\widehat W \in \mathcal{W}_\ell$ for which
    \[
        \tcw_\Theta(\widehat W) \geq \max_{W \in \mathcal{W}_\ell} \tcw_\Theta(W) - \eps
    \]
    from only the moments $M_1, \ldots, M_k$ of $\Theta$ for $k = 2 B d/\eps$, where $B$ is an upper bound on all $u_\theta(\phi)$.
\end{theorem}
The proof relies on establishing that any two distributions sharing the first $k$ moments are close in the 1-Wasserstein distance, which in turn bounds the estimation error for any Lipschitz-continuous welfare function.

\section{Discussion}
We study social choice problems when each voter provides only minimal preference feedback. In the linear social choice model—where candidates have vector embeddings and voter utilities are linear—such inference is possible. We analyze pairwise queries with and without intensity, formalizing the latter as a thresholded additive difference. 

We find a sharp divide between one comparison per voter and slightly richer elicitation. A single pairwise comparison suffices to identify and efficiently estimate the first moment (the “average voter”), enabling approximately welfare-optimal selection. However, the second moment is not identifiable from one comparison, so objectives that depend on dispersion or inequality, such as risk-adjusted and Nash welfare, cannot generally be supported by asking only a single query per voter.
On the other hand, with more comparisons per voter (or appropriately designed graded comparisons), higher-order structure becomes recoverable: with $k$ comparisons we can identify the first $k$ moments with polynomial sample complexity. Remarkably, two comparisons or one graded comparison per voter suffice to identify the full distribution, though with substantially weaker estimation efficiency.
These moment-recovery results enable principled social choice methods beyond social welfare maximization. For single-winner selection, estimating the first two moments supports risk-adjusted and Nash welfare objectives, while for multi-winner selection, moments can be used to approximately optimize coverage-style committee objectives.

Several challenges remain. Our analysis relies on geometric expressivity assumptions, making learnability under constrained query spaces an important open problem. Moreover, existing approximation bounds suggest that large $k$ may be required; developing objectives and algorithms that perform well with low-order moments and realistic sample sizes is key. Additionally, while using one size-1 graded query per vote to identify the distribution has theoretical guarantees, how to properly elicit and use preference intensity queries in practice is nontrivial.
Finally, extending moment-based approaches such as tensor decomposition~\cite{anandkumar2014tensor} and sum-of-squares methods~\cite{laurent2008sums} for more complex tasks is an interesting direction. 

More broadly, our work highlights a design principle for preference collection in social choice and beyond: small increases in per-user feedback richness can qualitatively expand what society-level properties are learnable, enabling alignment systems that capture not only average preferences but also disagreement and representation.



\bibliographystyle{plainnat}
\bibliography{refs.bib}

\appendix
\newpage
\section{Notation Reference}\label{app:notation}
\begin{table}[h]
    \centering
    \begin{tabular}{ll}
        \hline
        \textbf{Symbol} & \textbf{Description} \\
        \hline
        $\sdm$ & The unit sphere in $\mathbb{R}^d$\\
        $\bar{\sigma}$ & Uniform probability measure (normalized surface area measure) on $\mathbb{S}^{d-1}$\\
        $\Theta$ & The underlying population distribution of voter types over $\mathbb{S}^{d-1}$\\
        $x \in \mathcal{X}, y \in \mathcal{Y}$ & Prompts and Responses\\
        $\Phi(x, y)$ or $\phi$ & Embedding vector of a candidate (or prompt-response pair)\\
        $u_\theta(\phi)$ &Utility of voter $\theta$ for candidate $\phi$, given by $\theta \cdot \phi$\\
        $q$ & Pairwise comparison query vector (difference of embeddings $q \defeq \phi_1 - \phi_2$)\\
        $\resp_\theta(q)$ & Binary response of voter $\theta$ to query $q$ (e.g., $\1\{\theta \cdot q \ge 0\}$)\\
        $\mathrm{grad}_\theta(q)$ & Binary response of voter $\theta$ to query $q$ (e.g., $\1\{\theta \cdot q \ge \tau\}$)\\
        $Q_k(\bq; \Theta)$ or $\Q_k(\bq)$ & Probability that a random voter $\theta \sim \Theta$ responds positively to all queries  $\bq$\\
        $G_\tau(\bq; \Theta)$ or $G_\tau(\bq)$ & Probability that a random voter $\theta \sim \Theta$ has strong preference to $\bq$\\
        $c_d$ & The constant from the first moment identity (see \Cref{eq:cd-constant-def}) \\
       $Z_k$ & The set of $\tau \in (0,1)$ for which $G_{\tau}$ fails to identify $M_k$  \\\hline
    \end{tabular}
    \caption{Notation used throughout this work.}
    \label{tab:notation}
\end{table}

\section{Spherical Harmonics Basics}
A spherical harmonic of degree $j$ is the restriction to $\sdm$ of a degree-$j$ homogeneous polynomial whose Laplacian $\Delta \defeq \sum_{i=1}^d \partial_i^2$ vanishes. Spherical harmonics arise naturally in our analysis because the distributions we study are supported on the sphere, and moment functionals correspond to homogeneous polynomials. There are many excellent references on spherical harmonics, including the classical text~\cite{groemer1996geometric} and the more recent ~\cite{dai2013approximation}; we refer interested readers to these sources for further background.

Spherical harmonics enjoy many useful properties, including orthogonality and an $L^2$-decomposition into harmonic subspaces. The main tool deferred from the main text is the classic Funk–Hecke formula, which dates back to the work of \citet{funk1915beitrage} and~\citet{hecke1917orthogonal}.

\begin{theorem}[Funk-Hecke Theorem]
\label{thm:funk-hecke}
Suppose $d \ge 2$ and let $Y \in \mathcal{H}_j$, $x \in \sdm$, $K$ is bounded. Then, for all $y \in \sdm$,
\[
    \int_{S^{d - 1}} K(x\cdot y) Y(x) \; d\sigma(x) = \mu_j \cdot Y(y)
\]
where $\mu_j=\frac{\Gamma(\frac{d}{2})}{\sqrt{\pi}\Gamma(\frac{d-1}{2})} \int_{-1}^1 P_{j}(t) K(t) (1-t^2)^{\frac{d-3}{2}} dt$, where $P_{j}$ being the $j$-th Gegenbauer polynomial with parameter $\frac{d-2}{2}$.
\end{theorem}

Gegenbauer polynomials form a classical family of orthogonal polynomials on $[-1,1]$ and are closely connected to spherical harmonics. In the main text, we use the basic parity property: the polynomial is even when $j$ is even and odd when $j$ is odd. Further background and properties can be found in Appendix B of~\cite{dai2013approximation}.

\section{Omitted Results and Proofs} \label{app:proofs}
\subsection{Proof of Observation~\ref{thm:k-moment-gap}}
\begin{proof}
Having the same $k$-th moment tensor is the same as having same expected values for each monomials of degree $k$, and if the two distribution have the same first $k$ moments, for any polynomial with degree less or equal to $k$, their expected values over the distributions will also be the same by the linearity of the expected values.

    Let $\sigma$ denote the uniform probability measure on $\sph^{d-1}$.
By \cite{axler2001harmonic}[Proposition 5.9]
for any polynomial $p$ and any homogeneous harmonic polynomial $q$ on $\mathbb R^d$
satisfying $\deg(q)>\deg(p)$,
\[
\int_{\sph^{d-1}} p(x)\,q(x)\, d\sigma(x) = 0.
\tag{1}
\label{eq:harmonic}
\]
Choose any nonzero homogeneous harmonic polynomial of degree $k+1$; for instance
\[
h(x)=\Re\bigl( (x_1 + i x_2)^{k+1}\bigr), 
\]
where $x_1,x_2$ are the first and second coordinates of the vector $x$ and $\Re$ represents the real part of a complex number.

By the mean value property of a harmonic function, $ \int h(x) d\sigma(x)=h(0)=0$. Thus we are able to 
define
\[
d\mu_\pm(x) = (1\pm \varepsilon h(x))\, d\sigma(x)
\]
for some  sufficiently $\varepsilon>0$  such that
\[
1\pm \varepsilon h(x)\ge 0 \qquad \text{for all } x\in \sph^{d-1}.
\]

Let $x^\alpha$ be a monomial with $|\alpha|\le k$.
Applying Proposition~5.9 with $p(x)=x^\alpha$ and $q(x)=h(x)$ gives
\[
\int_{S^{d-1}} x^\alpha(x)\,h(x)\, d\sigma(x) = 0,
\]
since $\deg(h)=k+1 > |\alpha| = \deg(p)$.
Hence
\[
\int x^\alpha\, d\mu_+ - \int x^\alpha\, d\mu_-
= 2\varepsilon\int_{S^{d-1}} x^\alpha(x)\, h(x)\, d\sigma(x)
=0,
\]
so all moments up to order $k$ agree.

Because $h$ is not identically zero,
\[
\int_{S^{d-1}} h(x)^2\, d\sigma(x) > 0.
\]
Thus
\[
\int h\, d\mu_+ - \int h\, d\mu_-
= 2\varepsilon\int_{S^{d-1}} h(x)^2\, d\sigma(x) \neq 0.
\]
Since $h$ is a homogeneous polynomial of degree $k+1$, this implies that
the $(k+1)$-st moments of $\mu_+$ and $\mu_-$ differ.
\end{proof}

\subsection{Proof of \Cref{lem:uniqueness-of-Theta-from-moments}}
\begin{proof}[Sketch]
    Proofs of this sort proceed generally in three steps.
    The first is to argue that the moment tensor $M_k$ for a distribution $\mu$ suffices to compute the expectation $\E{\mu}{p(x)}$ of any homogeneous degree-$k$ polynomial $p$ on $\sdm$.
    
    The second step is invoke the Stone-Weierstrass theorem for $\sdm$, which says that the set of polynomials is dense in the larger space of continuous functions $C(\sdm)$.
    Here density means uniform convergence; for all continuous $f \in C(\sdm)$ there is a sequence of polynomials which converges to it in $\norm{f - p}_\infty = \max_{x \in \sdm} \abs{f(x) - p(x)}$.
    This is quite strong, and it is possible because the domain is compact. 
    
    The last step is to invoke the Riesz-Markov-Kakutani representation theorem, which states that any two measures that agree on all continuous functions $\E{}{f(x)}$ are in fact the same.
    Since the polynomials are dense in $C(\sdm)$, if $\mu$ and $\mu'$ agree on their moments (and therefore on $\E{}{p(x)}$ for all polynomials), then $\E{\mu}{f(x)} = \E{\mu'}{f(x)}$ for all continuous $f$ and so this representation theorem applies.
\end{proof}

\subsection{Proof of \Cref{thm:moment_identity_graded}}
To prove Theorem~\ref{thm:moment_identity_graded}, we need the following lemma. For a fixed voter direction $\theta \in \mathbb S^{d-1}$, consider its contribution \[
I_{k,\tau}(\theta)
= \int_{\mathbb S^{d-1}} \1\{\theta^\top q \ge \tau\}\, q^{\otimes k}\, d\bar \sigma(q).\]

\begin{lemma}\label{lem:threshold_identity}
Let $d \ge 2$, $k \ge 1$, for any  $\theta\in\sdm$, it holds that
\[
I_{k,\tau}(\theta)
= \sum_{j=0}^{\lfloor k/2\rfloor} \lambda_{k,j}\,
\mathrm{Sym}\bigl(\theta^{\otimes (k-2j)} \otimes I^{\otimes j}\bigr),
\] where the constants $\lambda_{k,0},\dots,\lambda_{k,\lfloor k/2\rfloor}$ depend on
$\tau,d,k$,  $I$ is the rank-2 identity tensor (i.e. matrix) on $\mathbb R^d$, and $\mathrm{Sym}$ denotes full symmetrization over all indices.

\end{lemma}
\begin{proof}

Let $x \in \mathbb{R}^d$ be an arbitrary vector. We define the scalar polynomial $p(x)$ by contracting $I_{k,\tau}(\theta)$ with $k$ copies of $x$:
\[
p(x) \defeq \langle I_k(\theta), x^{\otimes k} \rangle
= \left\langle \int_{\mathbb S^{d-1}} \1\{\theta^\top q \ge \tau\}\, q^{\otimes k}\, d\sigma(q), \; x^{\otimes k} \right\rangle.
\]
As $\langle q^{\otimes k}, x^{\otimes k} \rangle = (\langle q, x \rangle)^k$, we have:
\[
p(x) = \int_{\mathbb S^{d-1}} \1\{\theta^\top q \ge \tau\}\, \, (q^\top x)^k \, d\sigma(q).
\]

Let $O(d)$ denote the orthogonal group acting on $\mathbb{R}^d$. Consider the stabilizer subgroup of $\theta$, denoted $\mathcal{R}_\theta$.
For $R \in \mathcal{R}_\theta$:
\[
p(Rx) = \int_{\mathbb S^{d-1}} \1\{\theta^\top q \ge \tau\}\, \, (q^\top Rx)^k \, d\sigma(q).
\]
Now let $u = R^\top q$. Since $R$ is orthogonal, the measure is invariant, so $d\sigma(q) = d\sigma(u)$.
Furthermore, 
\(
\theta^\top q = \theta^\top (Ru) = (R^\top \theta)^\top u = \theta^\top u.
\)
Substituting these back into the integral:
\[
p(Rx) = \int_{\mathbb S^{d-1}} \1\{\theta^\top u \ge \tau\}\, \, (u^\top Rx)^k \, d\sigma(u) = p(x).
\]
Thus, $P(x)$ is a polynomial in $x$ that is invariant under all rotations about the axis $\theta$.

Therefore, $p(x)$ must be of the form:
\(
p(x) = F(\langle x, \theta \rangle, |x|^2).
\)
Since $p(x)$ is defined by an integral of $(q^\top x)^k$, $p(x)$ is a {homogeneous polynomial of degree $k$} in $x$.
Consequently, $F$ must be a linear combination of terms of the form $(\langle x, \theta \rangle)^a (|x|^2)^b$ such that the total degree matches $k$:
\(
a + 2b = k.
\)
Since $a, b$ must be non-negative integers, let $b=j$. Then $a = k-2j$. The possible values for $j$ are integers satisfying $0 \le 2j \le k$, i.e., $0 \le j \le \lfloor k/2 \rfloor$.
Thus, there exist scalar constants $\lambda_{k,j}$ such that:
\begin{equation}
    p(x) = \sum_{j=0}^{\lfloor k/2 \rfloor} \lambda_{k,j} \, (\langle x, \theta \rangle)^{k-2j} (|x|^2)^j.
    \label{eq:polynomial_thetaok}
\end{equation}

We now map the scalar terms back to their tensor equivalents.
\begin{itemize}
    \item The term $(\langle x, \theta \rangle)^{k-2j}$ corresponds to the contraction of the rank-$(k-2j)$ tensor $\theta^{\otimes (k-2j)}$ with $x^{\otimes (k-2j)}$.
    \item The term $|x|^{2j} = (\langle x, I x \rangle)^j$ corresponds to the contraction of $j$ copies of the identity matrix, $I^{\otimes j}$, with $x^{\otimes 2j}$.
\end{itemize}
The product corresponds to the tensor product $\theta^{\otimes (k-2j)} \otimes I^{\otimes j}$ contracted with $x^{\otimes k}$.
However, the original tensor $I_{k,\tau}(\theta)$ is fully symmetric, while the tensor product $\theta^{\otimes (k-2j)} \otimes I^{\otimes j}$ is not. Since the equality holds for all $x$, the symmetric tensors associated with the polynomials must be equal.
Finally, we apply the symmetrization operator $\mathrm{Sym}$ to the basis tensors:
\(
I_{k,\tau}(\theta) = \sum_{j=0}^{\lfloor k/2 \rfloor} \lambda_{k,j} \, \mathrm{Sym}\left( \theta^{\otimes (k-2j)} \otimes I^{\otimes j} \right).
\)
\end{proof}

\begin{proof}[Proof of Theorem~\ref{thm:moment_identity_graded}]
We first show that all moments are identifiable through induction on $k$.  Let $k \ge 1$ be fixed and assume that all moments $\mathbb{E}_{\theta\sim \Theta}[\theta^{\otimes \ell}]$ for $\ell < k$ are known. Now by Lemma~\ref{lem:threshold_identity} and Fubini's swap,
\begin{align*}
T_{k,\tau}\defeq&\int_{\mathbb{S}^{d-1}} G_\tau(q) \, q^{\otimes k} \, d\bar\sigma(q)\\
= &\int_{\mathbb{S}^{d-1}} \left( \int_{\mathbb{S}^{d-1}} \mathbf{1}(\theta^\top q \ge \tau) \, q^{\otimes k} \, d\bar\sigma(q) \right) d\Theta(\theta) \\
=& \mathbb{E}_{\theta\sim \Theta} [I_{k,\tau}(\theta)]\\
=& \lambda_{k,0}(\tau) \, \mathbb{E}[\theta^{\otimes k}] 
+ \sum_{j=1}^{\lfloor k/2 \rfloor} \lambda_{k,j}(\tau) \, 
\mathrm{Sym} \Bigl( \mathbb{E}_{\theta \sim \Theta}[\theta^{\otimes (k-2j)}] \otimes I^{\otimes j} \Bigr),
\label{eq:moment-decomp}
\end{align*}
The summation term involves only moments of order $k-2, k-4, \dots$, which are known by the inductive hypothesis. Therefore, the $k$th moment tensor $\mathbb{E}[\theta^{\otimes k}]$ is uniquely determined if and only if the coefficient $\lambda_{k,0}(\tau)$ is non-zero.

To rigorously determine $\lambda_{k,0}(\tau)$, we utilize spherical harmonics again. First, for a fixed $x$, the homogeneous polynomial $f(q)=(q^Tx)^k$ can be decomposed as a sum of spherical harmonics of degree $k,k-2, \dots$: $f(q)= \sum_{j=0}^{\lfloor k/2\rfloor} |x|^{2j} Y_{k-2j}(q)$. 
Now applying the linearity of the integral and  the Funk-Hecke Theorem (see Theorem~\ref{thm:funk-hecke} in Appendix~\ref{app:proofs}), 
we have that $$p(x)=\sum_{j=0}^{\lfloor k/2 \rfloor} \mu_{k-2j}\, |x|^{2j}Y_{k-2j}(\theta).$$
Comparing this with Equation~(\ref{eq:polynomial_thetaok}), we see that $\mu_k=\lambda_{k,0}.$
Thus \begin{equation}
    \lambda_{k,0}(\tau)=\frac{\Gamma(\frac{d}{2})}{\sqrt{\pi}\Gamma(\frac{d-1}{2})} \int_{\tau}^1 P_{k}(t)  (1-t^2)^{\frac{d-3}{2}} dt,
\end{equation}
where $P_{k}$ is the Gegenbauer polynomial with parameter $\frac{d-2}{2}$.
Note that the derivative of $\lambda_{k,0}(\tau)$ (specifically, $P_k(t)$) has $k$ distinct roots on $(-1,1)$, and $\lfloor k/2 \rfloor$ on $(0,1)$ by symmetry. So $\lambda_{k,0}(\tau)$ has $\lfloor k/2 \rfloor$ stationary points. Now consider the boundary points for $\lambda_{k,0}(\tau)$. Clearly, $\lambda_{k,0}(1)=0$ .When $\tau=0$ and $k$ is even, both $P_{k,d}(t)$ and $(1-t^2)^{(d-3)/2}$ are even and due to the orthogonality of Gegenbauer polynomials with parameter $\frac{d-2}{2}$, $\int_{0}^{1}P_{k}(t)(1-t^2)^{(d-3)/2}=\frac{1}{2}\int_{-1}^{1} P_{k}(t)\cdot 1 \cdot (1-t^2)^{(d-3)/2}=0$, so $\lambda_{k,0}(\tau)=0$ and the function has one less interior zero. Hence the number of roots of the function $\lambda_{k,0}(\tau)$ is $\lfloor \frac{k-1}{2} \rfloor$.
And $Z_k \defeq \{ \tau \in (0,1) : \lambda_{k,0}(\tau) = 0 \}$ must be discrete. Thus, for any $\tau \notin Z_k$, we can solve for the $k$th moment:
\begin{equation*}
\mathbb{E}[\theta^{\otimes k}] = \frac{1}{\lambda_{k,0}(\tau)} \left( T_{k,\tau} - \sum_{j=1}^{\lfloor k/2\rfloor} \lambda_{k,j}(\tau) \, \mathrm{Sym}\bigl(\mathbb{E}[\theta^{\otimes (k-2j)}] \otimes I^{\otimes j}\bigr) \right)
\end{equation*}

Now since the set $\bigcup_{k\ge 1} Z_k$ is countable, and has measure zero on $(0,1)$, for almost every $\tau \in (0,1)$, all moments are identifiable with $G_\tau$. \Cref{lem:uniqueness-of-Theta-from-moments} then implies that distributions are identifiable.  
\end{proof}

\subsection{Proof of \Cref{thm:sample-size-2}}
\begin{proof}
    Fix odd $\ell$ and let $i_1\ldots,i_\ell$ be a sequence of indices. Let $p(\theta) = \theta_{i_1} \cdot \cdots \cdot \theta_{i_\ell}$. Recall, from the proof of \Cref{thm:identifiability-from-two}, there are spherical harmonics $f_j$ for odd $j \le \ell$ such that $p(\theta) = \sum_{j \le \ell: j \text{ is odd}} f_j(\theta)$. However, we did not give an explicit construction of them. It turns out, that finding them is possible. The primary property we need is that in our case, $|f_j(\theta)| \le d^{j/2}$ for all $\theta \in \sdm$. We encapsulate this into the following lemma, whose proof we differ until later:
    \begin{lemma}\label{lem:bound-poly}
    Let $p(x) = x_{i_1}x_{i_2}\dots x_{i_\ell}$ be a monomial of degree $\ell$ in $\mathbb{R}^d$. Let $f_j \in \mathbb{Y}_j^d$ denote the $j$-th spherical harmonic component of the restriction of $p(x)$ to the unit sphere $\mathbb{S}^{d-1}$. Then, 
$$|f_j(x)| \le d^{j/2}$$ for all $x \in \sdm$.
\end{lemma}
\begin{proof}[Proof of \Cref{lem:bound-poly}] 
The key tool we will need is the \emph{projection operator} as a way to determine the $j$'th spherical harmonic in the decomposition.
\citet[Definition 2.11]{atkinson2012spherical} call this function $\mathcal{P}_{j,d}f)$.

The only property we will need is Equation (2.49) which shows that
\[
    \|\mathcal{P}_{j,d}f\|_{C(\sdm)} \le N^{1/2}_{j, d} \|f \|_{C(\sdm)}.
\]
where, by Equation (2.10), $
N_{j, d} = \binom{j+d-1}{j} - \binom{j+d-3}{j-2}$ and, $\|g\|_{C(\sdm)} = \sup\{|g(\xi)| : \xi \in \sdm\}$ (see Section 1.3). Note that trivially \[N_{j, d} \le \binom{j+d-1}{j} = \frac{j + d - 1}{j} \cdot \cdots \cdot \frac{d + 1 - 1}{1} \le d^j.\]

Plugging in our monomial, note that on $\sdm$, each coordinate $|\theta_i| \le 1$, therefore, $|p(\theta)| \le 1$. Combining these yields the lemma statement.
\end{proof}

Next, we also need an explicit construction of $\lambda_j$. From \citet[Lemma 3.4.6]{groemer1996geometric}, using our normalized measure $d\bar{\sigma}$, this is given by
\[
    \lambda_j = (-1)^{(d-1)/2} \frac{1}{d} \cdot \frac{1 \cdot 3 \cdot \cdots \cdot (j - 2)}{(d + 1)(d+3) \cdots (d+ j - 2)}.
\]
For our purposes, we will lower bound $|\lambda_j| \ge (d + 1)^{-(j + 1)/2}$.

Combining these, we see that
\begin{equation}\label{eq:bound-harmonic}
    \left|\sum_{j \le \ell : j \text{ is odd}} \frac{1}{\lambda_j} f_j(\theta)\right| \le \sum_{j \le \ell: j \text{ is odd}} (d + 1)^{j + 1/2} \le 2(d+1)^{\ell + 1/2}.
\end{equation}

Now we will apply this to our concentration inequalities.

We start with odd $k$. We will work with the matricized form where the entire $[k]$ is on one side, so this is essentially a $d^k \times 1$ sized vector. Define $\Psi^k(q)$ as the matricized vector where for a each index sequence $\alpha = i_1\ldots i_k$, $\Psi^k(q,b)_\alpha = b \cdot \sum_{j \le k : j \text{ is odd}} \frac{1}{\lambda_j}f^\alpha_j(q)$ where $f^\alpha_j$ is the spherical harmonic of degree $j$ in the decomposition of the monomial corresponding to $\alpha$. As we showed, $\E{q \sim \bar{\sigma}, \theta \sim \Theta}{\Psi^k(q, \resp_\theta(q))_\alpha}$ is the $\alpha$ entry of $M_k$. Thus, we will set our estimator $\widehat{M}_k$ to be the unmatricized version of $\frac{1}{T} \sum_{i = 1}^T \Psi^k(q^i, \resp_{\theta_i}(q^i))$ where $q^i$ is the $i$'th sampled query and $\resp_{\theta_i}(q^i)$ is the response of the sampled voter $\theta_i$. 

We would like to apply Matrix Bernstein to this. By Inequality~\eqref{eq:bound-harmonic}, each entry of $\Psi^k(q_i, \resp_{\theta_i}(q_i))$ is bounded by $2(d + 1)^{k + 1/2}$. The expected value must also lie in $[-2(d + 1)^{k + 1/2}, 2(d + 1)^{k + 1/2}]$, and thus the maximum distance from the expected value is at most $4(d + 1)^{k + 1/2}$. As there are $d^k$ entries, the $L_2$ norm (equivalent to the operator norm for vectors) is upper bounded by $d^{k/2} \cdot 4(d + 1)^{k + 1/2} \le 4(d+1)^{(3k+1)/2}$. As for the variance bound, note that for a vector $v$, both $\|v\tp v\|$ and $\|v v\tp\| $ are bounded by $\|v\|_2^2$. Thus, we can upperbound this by $16(d + 1)^{3k + 1}$. Plugging this into matrix Bernstein, we get
\[
    T \ge \frac{2}{\eps^2}\left(16(d + 1)^{3k + 1} + \frac{4(d+1)^{(3k+1)/2}\varepsilon}{3} \right)\log \left(\frac{d^k + 1}{\delta} \right) = O\left( \frac{(d + 1)^{3k + 1}}{\eps^2}\log\left(\frac{d^k}{\delta} \right)\right).
\]

Next, consider even $k$. Let $s = k/2 - 1$. We will matricize such that $|I| = s$ and $|J| = k - s$. In particular, we will set our estimator $\widehat M_k$ to be the unmatricized version of $$\frac{1}{T} \sum_{i = 1}^T \Psi^s(q^i_1, \resp_{\theta_i}(q^i_1))\tp \Psi^{k - s}(q^i_2, \resp_{\theta_i}(q^i_2).$$
Let $\alpha$ be a sequence of indices and let $\alpha^s$ and $\alpha^{k-s}$ be its first $s$ and last $k - s$ indices collectively. The $(\alpha^s, \alpha^{k-s})$ entry of this outer product (which corresponds to $\alpha$) will have expectation (by linearity) exactly the moment corresponding to $\alpha$.

We now apply matrix Berstein on these. Note that the individual $\Phi^s$ and $\Phi^{k-s}$ have terms bounded by $(d + 1)^{s + 1/2}$ and $(d + 1)^{(k - s) + 1/2}$ respectively. Thus, the distance each term is from the expectation is at most $2(d + 1)^{s + 1/2}$ and $2(d + 1)^{(k - s) + 1/2}$, respectively. Hence, their $L_2$ norms are at most $2(d+1)^{2s + 1/2}$ and $2(d + 1)^{2(k - s) + 1/2}$, respectively. Therefore, the spectral norm of their outer product is at most $2(d+1)^{2k + 1}$.

Next, we consider the variance. In general, we have two vectors $u$ and $v$ of dimensions $n$ and $n'$ and entries bounded by $L$ and $L'$, respectively, then $\|(u\tp v)\tp(u\tp v)\|$ and $\|(u\tp v)(u\tp v)\tp\|$ are bounded by $n \cdot n' \cdot L^2 \cdot (L')^2$. Plugging this in for use, we have a bound of $d^s \cdot d^{k - s} \cdot (2(d + 1)^{s + 1/2})^2 (2(d + 1)^{k - s + 1/2})^2 \le 16 (d + 1)^{3k + 2}$. Plugging this into matrix Berstein, we get
\[
    T \ge \frac{2}{\eps^2}\left(16(d + 1)^{3k + 2} + \frac{4(d+1)^{(3k+2)/2}\varepsilon}{3} \right)\log \left(\frac{d^{k/2 + 1} + d^{k/2 - 1}}{\delta} \right) = O\left( \frac{(d + 1)^{3k + 2}}{\eps^2} \log\left(\frac{d^k}{\delta} \right)\right). \qedhere
\]
\end{proof}

\subsection{Proof of~\Cref{prop:welfare_max}}
\begin{proof}[]
    By \Cref{thm:complexity-first}, $T$ samples suffice to estimate the first moment $M_1$ such that $\|\hat{M}_1 - M_1\|_{op} \le \varepsilon$. The estimated welfare for any candidate $\phi$ is $\langle \hat{M}_1, \phi \rangle$. The estimation error is $|\langle \hat{M}_1 - M_1, \phi \rangle| \le \|\hat{M}_1 - M_1\|_{op} \|\phi\| \le \varepsilon B$. If we select the candidate maximizing the estimated welfare, the true welfare of the selected candidate is at most $2\varepsilon B$ suboptimal.
\end{proof}
\subsection{Proof of \Cref{cor:risk_welfare}}
\begin{proof}
    It is without loss of generality to assume $\varepsilon \le 1$, as the 0-tensor is a trivial 1-approximation since $\|M_k\| \le 1$.
    Using \Cref{thm:mk_matrix_bernstein} with $k=2$, we obtain estimates $\hat{M}_1$ and $\hat{M}_2$ such that $\|\hat{M}_1 - M_1\| \le \varepsilon$ and $\|\hat{M}_2 - M_2\| \le \varepsilon$.
    Recall that $\raw(\phi) = M_1^\top \phi - \alpha \sqrt{\phi^\top M_2 \phi - (M_1^\top \phi)^2}$.
    Let $\mu = M_1^\top \phi$ and $\hat{\mu} = \hat{M}_1^\top \phi$. Similarly, let $\sigma^2 = \phi^\top M_2 \phi - \mu^2$ and $\hat{\sigma}^2 = \phi^\top \hat{M}_2 \phi - \hat{\mu}^2$.
    We have $|\hat{\mu} - \mu| \le \varepsilon B$ and $|\phi^\top (\hat{M}_2 - M_2) \phi| \le \varepsilon B^2$.
    The error in the variance term is bounded by:
    \[ |\hat{\sigma}^2 - \sigma^2| \le \varepsilon B^2 + |\hat{\mu}^2 - \mu^2| \le \varepsilon B^2 + 2B(\varepsilon B) \le 3\varepsilon B^2. \]
    Using the fact that $|\sqrt{x} - \sqrt{y}| \le \sqrt{|x-y|}$, the error in the standard deviation is bounded by $\sqrt{3\varepsilon} B$.
    Thus, the total estimation error is bounded by $\varepsilon B + \alpha \sqrt{3\varepsilon} B \le \sqrt{3} \cdot (\alpha + 1) B \sqrt{\varepsilon}$ because $\varepsilon \le 1$. If we select the candidate maximizing estimated risk-adjusted welfare, the true risk-adjusted welfare of the selected candidate is at most double.
\end{proof}

\subsection{Missing Portion of Proof of \Cref{thm:nash_approx}}
Here we derive \eqref{eq:log-as-chebyshev-sum}.
Recall that we will make use of the identity~\cite[Equation 1.514]{gradshteyn2014table}
\[
    \log(1 - 2 \alpha \cos(\varphi) + \alpha^2) = -2\sum_{j = 1}^\infty \frac{\cos (j\varphi)}{j} \cdot \alpha^j,
\]
for $\alpha^2 \le 1$ and $\alpha \cos(\varphi) \ne 1$.
Our goal is to show that
\[
    \log(1 + \delta t) = -2\sum_{j = 1}^\infty \frac{\alpha^j}{j} \cdot T_j(t),
\]
where $\delta = \frac{r - 1}{r + 1}$. 

As long as $|\alpha| < 1$, for real $\varphi$, $\alpha\cos(\varphi) \ne 1$. By our definition of Chebyshev polynomials, this implies that for $t \in [-1, 1]$,
\[
    \log(1 - 2 \alpha t + \alpha^2) = -2\sum_{j = 1}^\infty \frac{\alpha^j}{j} \cdot T_j(t).
\]
To get it in our form, we can normalize by $1 + \alpha^2$ to get 
\[
    \log\left(1 +  \frac{-2 \alpha}{1 + \alpha^2} t\right) = -\log(1 + \alpha^2) -2\sum_{j = 1}^\infty \frac{\alpha^j}{j} \cdot T_j(t).
\]
We will set $\alpha = - \frac{\sqrt{r} - 1}{\sqrt{r} + 1}$, which clearly has $|\alpha| < 1$ and yields
\[
    \frac{-2 \alpha}{1 + \alpha^2} = \frac{2 \cdot \frac{\sqrt{r} - 1}{\sqrt{r} + 1}}{1 + \left(\frac{\sqrt{r} - 1}{\sqrt{r} + 1} \right)^2} = \frac{2 \cdot \frac{\sqrt{r} - 1}{\sqrt{r} + 1}}{\frac{2(r + 1)}{(\sqrt{r} + 1)^2}} = \frac{(\sqrt{r} - 1)(\sqrt{r} + 1)}{r + 1} = \frac{r - 1}{r + 1} = \delta.
\]
First, observe that $\alpha^2 < 1$. This also implies that $\alpha \cos(\varphi) \ne 1$ for real $\varphi$. Thus, for $t \in [-1, 1]$,
\[
    \log(1 + \delta) = \log(1 - 2 \alpha x + \alpha^2) = -2\sum_{j = 1}^\infty \frac{\alpha^j}{j} \cdot T_j(x). \qed
\]

\subsection{Proof of \Cref{thm:k-medians-approximable-from-moments}}
We now prove a key general result which shows that any (welfare) function that is Lipschitz continuous can be approximated using the first $k$ moments.
This result relies on the following lemma, which can be understood as a counterpart to Observation~\ref{thm:k-moment-gap}. 
It quantifies the extent to which two measures that agree on their first $k$ moments can substantively differ. 
Here the difference between distributions $\mu$ and $\nu$ is measured according to the \emph{1-Wasserstein distance}, or earth-mover's distance, given by
\[
    W_1(\mu,\nu) \defeq \inf_{\gamma \in \Gamma} \iint_{\sdm \times \sdm} \norm{\theta - \theta'}_2 \: d\mu(\theta) \: d\nu(\theta'),
\]
where $\Gamma = \Gamma(\mu, \nu)$ is the set of all statistical couplings of $\mu$ and $\nu$.
\begin{lemma}[Wasserstein Bound via Moment Matching]
\label{lem:wasserstein}
    Let $\mu$ and $\nu$ be two probability measures on $\sdm \subset \mathbb{R}^{d}$. 
    Suppose the first $k$ moments of $\mu$ and $\nu$ are equal.
    Then the 1-Wasserstein distance satisfies
    \[
        W_1(\mu, \nu) \le C \cdot \frac{d}{k},
    \]
    where $C > 0$ is an absolute constant independent of the dimension $d$, the degree $k$, and the measures.
\end{lemma}
\begin{proof}
    We will use the fact that the 1-Wasserstein distance admits the following alternative definition via Kantorovich-Rubinstein duality:
    \begin{equation}
    \label{eq:dual-def-of-wasserstein-distance}
        W_1(\mu, \nu) = \sup_{f \in \Lip(\sdm)} \left| \int_{\sdm} f \, d\mu - \int_{\sdm} f \, d\nu \right|,
    \end{equation}
    where $\Lip(\sdm)$ is the set of functions on $\sdm$ that are 1-Lipschitz with respect to $L_2$. 
    For the purposes of this proof it will be more convenient to work with the \emph{geodesic distance}, given by $d_g(\theta, \theta') = \arccos(\langle \theta, \theta' \rangle)$ between $\theta, \theta' \in \sdm$. 
    They are equivalent for our purposes because $\norm{\theta, \theta'}_2 \leq d_g(\theta, \theta') \leq \frac{\pi}{2} \norm{\theta, \theta'}_2$ on the unit sphere.
    
    To begin, fix $f \in \Lip(\sdm)$. 
    Our goal is to approximate $f$ using a spherical polynomial $P_k \in \Pi_k^{d}$, where $\Pi_k^{d}$ denotes the space of all spherical polynomials of degree at most $k$. 
    By the triangle inequality,
    \begin{align}
        \abs{\E{\mu}{f} - \E{\nu}{f}}  &= \abs{\E{\mu}{f - P_k + P_k} - \E{\nu}{f - P_k + P_k}} \notag \\
        &\leq \abs{\E{\mu}{f - P_k}} + \abs{ \E{\nu}{P_k} - \E{\mu}{P_k}} + \abs{ \E{\nu}{f - P_k}} \notag \\
        &\leq \abs{\E{\mu}{f - P_k}} + \abs{\E{\nu}{f - P_k}}. \label{eq:norm-on-fxn-difference-bound}
    \end{align}
    This last step follows since $\E{}{P_k}$ can be written as a function of the at-most-$k$'th moments, which are equal for $\mu$ and $\nu$ by assumption. 
    
    We will bound the remaining two terms by the uniform norm over $\sdm$. 
    To this end, let $E_k(f) \defeq \min_{P \in \Pi_k^{d}} \norm{f - P}_\infty$ be the best possible uniform approximation to $f$ achievable using degree-$k$ polynomials.
    Choosing $P_k$ to be the polynomial approximation of $f$ witnessing this value, we have
    \begin{equation}
        \label{eq:poly-approx-intermediate-linfty-bound}
        \left| \E{\mu}{f - P_k} \right| \leq \int_{\sdm} \|f - P_k\|_\infty \; d\mu = E_k(f),
    \end{equation}
    and likewise for $\nu$. 
    Therefore from \eqref{eq:norm-on-fxn-difference-bound} we have
    \[
        \abs{ \E{\mu}{f} - \E{\nu}{f} } \le 2 \cdot E_k(f).
    \]
    To bound $E_k(f)$ we invoke the so-called Jackson theorem for the sphere, which relates the approximability of a function by polynomials to its smoothness.
    As established by \citet{newman1964jackson} (see also \cite[Lemma 41]{chen25learning}), for any continuous function $f$, the approximation error is bounded by the \emph{modulus of continuity}, which on $\sdm$ is given by $\omega(f, \delta)_\infty \defeq \sup_{x,y :\: d_g(x,y)\leq \delta} \abs{f(x) - f(y)}$.
    In particular, it holds that
    \begin{equation}
    \label{eq:jackson-bound}
        E_k(f) \le C_{NS} \cdot \omega\left(f, \frac{d}{k}\right)_\infty
    \end{equation}
    for some absolute constant $C_{NS}$ independent of the dimension $d$.
    Since $f$ is 1-Lipschitz, its modulus of continuity satisfies $\omega(f, t)_\infty \leq t$. 
    Applying this to \eqref{eq:jackson-bound} and substituting it into \eqref{eq:poly-approx-intermediate-linfty-bound} yields
    \[
        \left| \E{\mu}{f - P_k} \right| \leq 2\cdot C_{NS}\cdot \frac{d}{k}.
    \]
    Taking the supremum over all $1$-Lipschitz $f$, from \eqref{eq:dual-def-of-wasserstein-distance} we finally have
    \[
        W_1(\mu, \nu) \leq 2 \cdot C_{NS} \cdot \frac{d}{k}. \qedhere
    \]
\end{proof}

Now, we are ready to present the general result.
We do not claim that the resulting dependence of the required number of moments on $\eps$ is optimal, opting instead for a readily generalizable approach.

\begin{lemma}
\label{thm:convergence-in-distribution-from-moments}
    Let $\mu$ and $\nu$ be probability measures on the unit sphere $\sdm \subset \mathbb{R}^{d}$ such that their first $k$ moments match. 
    For any $L$-Lipschitz function $f: \sdm \to \mathbb{R}$, the integration error is bounded by
    \[
        \left| \int_{\sdm} f \, d\mu - \int_{\sdm} f \, d\nu \right| \le C \cdot L \cdot \frac{d}{k}, 
    \]
    where $C > 0$ is an absolute constant independent of the dimension $d$.
\end{lemma}
\begin{proof}
    This follows directly from \Cref{lem:wasserstein}.
    Given some $L$-Lipschitz $f$, observe that $g \defeq f/L$ is 1-Lipschitz. 
    By the dual description of $W_1(\mu, \nu)$ \eqref{eq:dual-def-of-wasserstein-distance}, we therefore have
    \[ 
        \left| \mathbb{E}_\mu[g] - \mathbb{E}_\nu[g] \right| \le W_1(\mu,\nu) \leq 2C \cdot \frac{d}{k}.
    \]
    Multiplying both sides by $L$, by linearity of expectation we obtain
    \[ 
        \left| \mathbb{E}_\mu[f] - \mathbb{E}_\nu[f] \right| \le 2C \cdot L \cdot \frac{d}{k}.\qedhere 
    \]
\end{proof}

We now use these lemmas to prove our main result on top-choice welfare.
\begin{proof}[Proof of \Cref{thm:k-medians-approximable-from-moments}]
    This proof proceeds from \Cref{thm:convergence-in-distribution-from-moments}. 
    We will apply it to the functions $\tc_W(\theta) \defeq \max_{\phi \in W} u_\theta(\phi)$.
    Observe that in our model all such $\tc_W(\theta)$ are $B$-Lipschitz in $\theta$, where $B$ is an upper bound on $\norm{\phi}$ for all $\phi \in W$. 

    Given the moments $M_1, \ldots, M_k$ of $\Theta$, construct $\widehat \Theta$ to be an arbitrary distribution over $\sdm$ consistent with these moments and choose $\widehat W \defeq \arg \max_{W \in \mathcal{W}_\ell} \tcw_{\widehat \Theta}(W)$. 
    Let $W^*$ be the (unknown) $\ell$-$\tcw$ optimum for $\Theta$.
    
    Then letting $\eps' = C \cdot B \cdot \frac{d}{k}$, by the optimality of $\widehat \Theta$  and applying \Cref{thm:convergence-in-distribution-from-moments} twice we have
    \begin{align*}
        \tcw_{\Theta}(\widehat W) &\geq \tcw_{\widehat \Theta}(\widehat W) - \eps' \\
        &\geq \tcw_{\widehat \Theta}(W^*) - \eps' \\
        &\geq \tcw_{\Theta}(W^*) - 2 \eps'
    \end{align*}
    Setting $\eps = 2 \eps'$ and solving for $k$, we have $k = 2 B d/\eps$, as claimed.
\end{proof}

\section{Extending to Stochastic Responses}
\label{app:generalized-responses}

Instead of assuming voters respond deterministically with $\resp_\theta(q) = \1\{\theta \cdot q \ge 0\}$, we now allow responses to be stochastic.

We consider the random utility model (RUM)~\cite{azari2012random}, which is a standard model in both social choice and alignment. Concretely, we can capture the stochasticity by
introducing a generic function relating each single voter's utility difference to their probability of response
\begin{equation*}
\label{eq:generic-response}
    \psi : \R \to [0,1]
    \quad\text{with}\quad
    \psi(t) + \psi(-t) = 1 \quad\text{for all } t \in \mathbb{R}.
\end{equation*}
 Then, when presented with a comparison $(x, y_1, y_2)$, a voter responds $1$ with probability $\psi(u_\theta(x, y_1) - u_\theta(x, y_2)) = \psi(\theta \cdot q)$ for $q = \Psi(x, y_1) - \Psi(x, y_2)$. More formally, given a voter type $\theta$,
\[
    \Pr[\resp_\theta(q) = 1 | \ \theta] = \psi(\theta \tp q).
\]
A common choice of the model is Bradley-Terry: $\psi_{\mathrm{BT}}(t) = \frac{1}{1 + \exp(-t)}$. 

In general, we can now generalize multi-query responses for independent $q_1,\dots, q_t$ as $$\Tilde Q_t (\bq) = \Pr_{\theta\sim\Theta}\left[ [\resp_\theta(q_1) = 1] \wedge \ldots \wedge [\resp_\theta(q_t) = 1]\right] = \E{\theta \sim \Theta}{\prod_{i=1}^t  \Pr[\resp_\theta(q_i) = 1 | \ \theta]}$$

First, consider the problem of selecting a social-welfare–maximizing candidate. The first step of Lemma~\ref{lem:first_identity} requires rotational equivariance of inner-product–based responses. This argument continues to hold for the stochastic model, since for any rotation $R$ we have $\psi((R\theta) \tp q)= \psi(\theta \tp (R^{-1}q)). $

\begin{lemma}\label{lem:first-identity-generalized}
Let $\psi:[-1,1]\to\mathbb [0,1]$, and $\theta\in \sdm$ be fixed. Then  
\[
   I_\psi(\theta)
   \defeq \int_{\mathbb S^{d-1}}
        \psi( \theta \tp q )\, q\, d\bar \sigma(q)= c_d(\psi)\ \theta,
\]

where the scalar $c_d(\psi)$ is given by
\[
   c_d(\psi)
   = 
     \frac{\Gamma\!\left(\frac{d}{2}\right)}{\sqrt{\pi}\,\Gamma\!\left(\frac{d-1}{2}\right)}
\int_{-1}^1 \psi(t)\,t\,(1-t^2)^{\frac{d-3}{2}}\,dt
\]
\end{lemma}

\begin{proof}

Let $R$ be any rotation in $\mathbb R^d$.  
Because $\langle R\theta,Rq\rangle = \langle\theta,q\rangle$ and
$d\bar \sigma$ is rotation equivariant,
\[
   I_\psi(R\theta)
   = \int \psi(\langle R\theta,q\rangle)\, q\, d\bar \sigma(q)
   = \int \psi(\langle\theta,R^{-1}q\rangle)\, q\, d\bar \sigma(q)
   = R \int \psi(\langle\theta,u\rangle)\, u\, d\bar \sigma(u)
   = R\, I_\psi(\theta).
\]

Let $\mathcal R_\theta$ be the group of rotations that keep $\theta$
fixed. For any $R\in\mathcal R_\theta$, we have $R\theta=\theta$, and
therefore by the equivariance established above,
\[
   R\,I_\psi(\theta) = I\psi(R\theta) = I_\psi(\theta).
\]

Hence $I_\psi(\theta)$ is the fixed point of the linear action of the group $\mathcal R_\theta$.
However, as $\mathcal R_\theta$ is acting as the full rotation group on the orthogonal complement of $\theta$, the only fixed point in $\theta^{\perp}$ is the zero vector, and the only dimension left is for the span of the vector itself:
\[
     I_\psi(\theta)=c_d(\psi)\,\theta + 0, \text{  for some scalar } c_d(\psi).
\]

Now to calculate the constant, we set $\theta=e_d$. 
Then
\begin{align*}
  c_d(\psi)
  & =  \int_{\mathbb S^{d-1}}
         \psi(q_d)\, q_d\, d\bar \sigma(q) \\
&=  \frac{\Gamma\!\left(\frac{d}{2}\right)}{\sqrt{\pi}\,\Gamma\!\left(\frac{d-1}{2}\right)}
\int_{-1}^1 \psi(t)\,t\,(1-t^2)^{\frac{d-3}{2}}\,dt.        
\end{align*} \end{proof}

\begin{observation}
   To improve the signal-to-noise ratio in our moment estimators,  we would like $c_d(\psi)$ to be as big as possible. Given that $\psi(x)+\psi(-x)=1 $,
\begin{align*}
  & \int_{-1}^1 \psi(t)\,t\,(1-t^2)^{\frac{d-3}{2}}\,dt \\
&= \int_{0}^1 \psi(t)\,t\,(1-t^2)^{\frac{d-3}{2}}\,dt + \int_{-1}^0 \psi(t)\,t\,(1-t^2)^{\frac{d-3}{2}}\,dt \\
&= \int_{0}^1 \psi(t)\,t\,(1-t^2)^{\frac{d-3}{2}}\,dt + \int_{-1}^0 (1-\psi(-t))\,t\,(1-t^2)^{\frac{d-3}{2}}\,dt \\
 &       =  \int_{0}^{1}
         [2\psi(t)-1]\, t\, (1-t^2)^{\frac{d-3}{2}} dt\\
         \end{align*}
   
Since $t\, (1-t^2)^{\frac{d-3}{2}}$ is always nonnegative on $[0,1]$, maximizing $c_d$ is the same as maximizing $\psi(t)$ point-wise on $t\in (0,1],$ which indicates that the more deterministic the preference is, the stronger is the signal.
\end{observation}

The remainders of the proofs for \Cref{lem:first-moment-identification} and \Cref{thm:welfare-max-from-pairwise-queries} carry over without modification. 
\begin{theorem}
    The welfare-maximizing candidate is identifiable from $\Tilde{Q}_1$. 
\end{theorem}
And since we can draw i.i.d. voters from the distribution $\Theta$, $\Tilde{Q}_1$ can be estimated as before.

Next, consider the task of estimating higher moments. \Cref{thm:kth-moment-identification-from-k} can be extended due to \Cref{lem:first-identity-generalized}.

\begin{theorem}
    The $k$th moment $M_k$ is identifiable with $\Tilde{Q}_k$.
\end{theorem}

And due to the properties of spherical harmonics including Funk-Hecke (\Cref{thm:funk-hecke}),
\begin{theorem}
    Distributions are identifiable with $\Tilde{Q}_2$.
\end{theorem}

\end{document}